\documentclass{article}

\PassOptionsToPackage{numbers,compress}{natbib}

\usepackage[final]{neurips_2018}

\usepackage[utf8]{inputenc} 
\usepackage[T1]{fontenc}    
\usepackage{hyperref}       
\usepackage{url}            
\usepackage{booktabs}       
\usepackage{amsfonts}       
\usepackage{nicefrac}       
\usepackage{microtype}      

\usepackage{xr}

\usepackage{graphicx}
\usepackage{color}
\usepackage{soul}
\usepackage{mathbbol,bm,amsmath,amssymb,amsthm}

\newcommand{\R}{\mathbb{R}}
\newcommand{\C}{\mathbb{C}}
\newcommand{\U}{\mathbb{1}}

\newcommand{\N}{\mathcal{N}}

\newcommand{\I}{\mathbb{1}}
\renewcommand{\L}{\mathcal{L}}
\newcommand{\F}{\mathcal{F}}
\newcommand{\D}{\mathcal{D}}
\renewcommand{\I}{\mathcal{I}}

\newcommand{\E}{\mathbb{E}}
\renewcommand{\O}{\mathcal{O}}
\newcommand{\nw}{^{\text{(new)}}}
\newcommand{\od}{^{\text{(old)}}}

\newcommand{\uh}{\hat u}
\newcommand{\uhj}{\hat u_j}
\newcommand{\ut}{\hat u_\tau}
\newcommand{\utj}{\hat u_{\tau,j}}

\newcommand{\xh}{\hat x}
\newcommand{\xhj}{\hat x_j}
\newcommand{\xt}{\hat x_\tau}
\newcommand{\xtj}{\hat x_{\tau,j}}
\newcommand{\xr}{\hat x_{\text{ridge}}}

\newcommand{\Htau}{H^*_\tau}

\DeclareMathOperator*{\argmin}{argmin}

\DeclareMathOperator{\sgn}{sgn}

\DeclareMathOperator{\erfc}{erfc}
\DeclareMathOperator{\erfcx}{erfcx}
\DeclareMathOperator{\tr}{tr}

\newtheorem{thm}{Theorem}
\newtheorem{prop}{Proposition}
\newtheorem{lem}{Lemma}

\title{Analytic solution and stationary phase approximation for the Bayesian lasso and elastic net}

\author{
  Tom Michoel\\
  The Roslin Institute, The University of Edinburgh, UK\\
  Computational Biology Unit,
  Department of Informatics, 
  University of Bergen,
  Norway\\
  \texttt{tom.michoel@uib.no} \\
}

\begin{document}

\maketitle

\begin{abstract}
  The lasso and elastic net linear regression models impose a double-exponential prior distribution on the model parameters to achieve   regression shrinkage and variable selection,  allowing the inference of robust models from large data sets.  However, there has been limited success in deriving estimates for the full posterior distribution of regression coefficients in these models, due to a need to evaluate analytically intractable partition function integrals. Here, the Fourier transform is used to express these integrals as complex-valued oscillatory integrals over ``regression frequencies''. This results in an analytic expansion and stationary phase approximation for the partition functions of the Bayesian lasso and elastic net, where the non-differentiability of the double-exponential prior has so far eluded such an approach. Use of this approximation leads to highly accurate numerical estimates for the expectation values and marginal posterior distributions of the regression coefficients, and allows for Bayesian inference of much higher dimensional models than previously possible.
\end{abstract}

\section{Introduction}
\label{sec:introduction}

Statistical modelling of high-dimensional data sets where the number of variables exceeds the number of experimental samples may result in over-fitted models that do not generalize well to unseen data. Prediction accuracy in these situations can often be improved by shrinking regression coefficients towards zero \cite{friedman2001elements}. Bayesian methods achieve this by imposing a prior distribution on the regression coefficients whose mass is concentrated around zero. For linear regression, the most popular methods are ridge regression \cite{hoerl1970ridge}, which has a normally distributed prior; lasso regression \cite{tibshirani1996regression}, which has a double-exponential or Laplace distribution prior; and elastic net regression \cite{zou2005regularization}, whose prior interpolates between the lasso and ridge priors. The lasso and elastic net are of particular interest, because in their maximum-likelihood solutions, a subset of regression coefficients are exactly zero. However, maximum-likelihood solutions only provide a point estimate for the regression coefficients. A fully Bayesian treatment that takes into account uncertainty due to data noise and limited sample size, and provides posterior distributions and confidence intervals, is therefore of great interest. 

Unsurprisingly, Bayesian inference for the lasso and elastic net involves analytically intractable partition function integrals and requires the use of numerical Gibbs sampling techniques \cite{park2008bayesian, hans2009bayesian, li2010bayesian, hans2011elastic}. However, Gibbs sampling is computationally expensive and, particularly in high-dimensional settings, convergence may be slow and difficult to assess or remedy   \cite{liu2004,mallick2013bayesian, rajaratnam2015fast, rajaratnam2015mcmc}. An alternative to Gibbs sampling for Bayesian inference is to use asymptotic approximations to the intractable integrals based on Laplace's method \cite{kass1989approximate, rue2009approximate}. However, the Laplace approximation requires twice differentiable log-likelihood functions, and cannot be applied to the lasso and elastic net models as they contain a non-differentiable term proportional to the sum of absolute values (i.e. $\ell_1$-norm) of the regression coefficients. 

Alternatives to the Laplace approximation have been considered for statistical models where the Fisher information matrix is singular, and no asymptotic approximation using normal distributions is feasible \cite{watanabe2013widely,drton2017bayesian}. However, in $\ell_1$-penalized models, the singularity originates from the prior distributions on the model parameters, and the Fisher information matrix remains positive definite. Here we show that in such models, approximate Bayesian inference is in fact possible using a Laplace-like approximation, more precisely the stationary phase or saddle point approximation for complex-valued oscillatory integrals \cite{wong2001asymptotic}. This is achieved by rewriting the partition function integrals in terms of ``frequencies''  instead of regression coefficients, through the use of the Fourier transform. The appearance of the Fourier transform in this context should not come as a big surprise. The stationary phase approximation can be used to obtain or invert characteristic functions, which are of course Fourier transforms \cite{daniels1954saddlepoint}. More to the point of this paper, there is an intimate connection between the Fourier transform of the exponential of a convex function and the Legendre-Fenchel transform of that convex function, which plays a fundamental role in physics by linking microscopic statistical mechanics to macroscopic thermodynamics and quantum to classical mechanics \cite{litvinov2005maslov}. In particular, convex duality \cite{boyd2004convex,rockafellar1997convex}, which maps the solution of a convex optimization problem to that of its dual, is essentially equivalent to writing the partition function of a Gibbs probability distribution in coordinate or frequency space (Appendix~\ref{sec:part-funct-four}).

Convex duality principles have been essential to characterize analytical properties of the max\-imum-like\-li\-hood solutions of the lasso and elastic net regression models \cite{osborne2000lasso, osborne2000new, elghaoui2012,tibshirani2012strong, tibshirani2013lasso,michoel2016}. This paper shows that equally powerful duality principles exist to study Bayesian inference problems.

\section{Analytic results}
\label{sec:analytic-results}

We consider the usual setup for linear regression where there are $n$ observations of $p$ predictor variables and one response variable, and the effects of the predictors on the response are to be determined by minimizing the least squares cost function $\| y - Ax\|^2$ subject to additional constraints, where $y\in\R^n$ are the response data, $A\in\R^{n\times p}$ are the predictor data, $x\in\R^p$ are the regression coefficients which need to be estimated and $\|v\|=(\sum_{i=1}^n |v_i|^2)^{1/2}$ is the $\ell^2$-norm. Without loss of generality, it is assumed that the response and predictors are centred and standardized, 
\begin{equation}\label{eq:stand}
  \sum_{i=1}^n y_i = \sum_{i=1}^n A_{ij} = 0 \quad\text{and}\quad
  \sum_{i=1}^n y_i^2 = \sum_{i=1}^n A_{ij}^2 = n \quad\text{for } j\in\{1,2,\dots,p\}.
\end{equation}
In a Bayesian setting, a hierarchical model is assumed where each sample $y_i$ is drawn independently from a normal distribution with mean $A_{i\bullet} x$ and variance $\sigma^2$, where $A_{i\bullet}$ denotes the $i^{\text{th}}$ row of $A$, or more succintly,
\begin{align}
  p(y\mid A,x) &= \N(Ax,\sigma^2\U),\label{eq:31}
\end{align}
where $\N$ denotes a multivariate normal distribution, and the regression coefficients $x$ are assumed to have a prior distribution
\begin{align}
  p(x) &\propto \exp\Bigl[-\frac{n}{\sigma^2}\bigl( \lambda \|x\|^2 + 2\mu
      \|x\|_1 \bigr) \Bigr],\label{eq:32}
\end{align}
where  $\|x\|_1=\sum_{j=1}^p |x_j|$ is the $\ell_1$-norm, and the prior distribution is defined upto a normalization constant. The apparent dependence of the prior distribution on the data via the dimension paramater $n$ only serves to simplify notation, allowing the posterior distribution of the regression coefficients to be written, using Bayes' theorem, as
\begin{equation}\label{eq:post}
  p(x\mid y,A) \propto p(y\mid x,A) p(x) \propto e^{-\frac{n}{\sigma^2} \L(x\mid y,A)},
\end{equation}
where
\begin{align}
  \L (x\mid y,A) &= \frac1{2n}\| y - Ax\|^2 + \lambda \|x\|^2 + 2\mu \|x\|_1 \label{eq:1a}\\
  &= x^T \bigl(\frac{A^TA}{2n} + \lambda\U\bigr)x -2  \bigl(\frac{A^Ty}{2n}\bigr)^Tx + 2\mu \|x\|_1+ \frac1{2n}\|y\|^2\label{eq:1}
\end{align}
is minus the posterior log-likelihood function. The maximum-likelihood solutions of the lasso ($\lambda=0$) and elastic net ($\lambda>0$) models are obtained by minimizing $\L$, where the relative scaling of the penalty parameters to the sample size $n$ corresponds to the notational conventions of \cite{friedman2010regularization}\footnote{To be precise, in \cite{friedman2010regularization} the penalty term is written as $\tilde \lambda (\frac{1-\alpha}{2} \|x\|_2^2 + \alpha \|x\|_1)$, wich is obtained from \eqref{eq:1a} by setting $\tilde \lambda=2(\lambda+\mu)$ and $\alpha=\frac{\mu}{\lambda+\mu}$.}. In the current setup, it is assumed that the parameters $\lambda\geq0$, $\mu>0$ and $\sigma^2>0$ are given a priori.

To facilitate notation, a slightly more general class of cost functions is defined as \begin{equation}\label{eq:2}
  H(x\mid C,w,\mu) = x^TCx - 2w^Tx +2\mu \|x\|_1,
\end{equation}
where $C\in\R^{p\times p}$ is a positive-definite matrix, $w\in\R^p$ is an arbitrary vector and $\mu>0$. After discarding a constant term, 
$\L(x\mid y,A)$ is of this form, as is the so-called ``non-naive'' elastic net, where $C=(\frac1{2n}A^TA+\lambda\U)/(\lambda+1)$ 
\cite{zou2005regularization}.  More importantly perhaps, eq.~\eqref{eq:2} also covers linear mixed models, where samples need not be independent \cite{rakitsch2012lasso}. In this case, eq.~\eqref{eq:31} is replaced by $p(y\mid A,x) = \N(Ax,\sigma^2K)$,
for some covariance matrix $K\in\R^{n\times n}$, resulting in a posterior minus log-likelihood function with $C=\frac1{2n}A^TK^{-1}A+\lambda\U$ and $w=\frac1{2n}A^TK^{-1}y$. The requirement that $C$ is positive definite, and hence invertible, implies that $H$ is strictly convex and hence has a unique minimizer. For the lasso ($\lambda=0$) this only holds without further assumptions if $n\geq p$ \cite{tibshirani2013lasso}; for the elastic net ($\lambda>0$) there is no such constraint. 

The Gibbs distribution on $\R^p$ for the cost function $H(x\mid C,w,\mu)$ with inverse temperature $\tau$ is defined as
\begin{equation*}
  p(x\mid C,w,\mu) = \frac{e^{-\tau H(x\mid C,w,\mu)}}{Z(C,w,\mu)}.
\end{equation*}
For ease of notation we will henceforth drop explicit reference to $C$, $w$ and $\mu$.  The normalization constant $Z = \int_{\R^p} e^{-\tau H(x)} dx$ is called the partition function. There is no known analytic solution for the partition function integral. However, in the posterior distribution \eqref{eq:post}, the inverse temperature $\tau=\frac{n}{\sigma^2}$ is large, firstly because we are interested in high-dimensional problems where $n$ is large (even if it may be small compared to $p$), and secondly because we assume a priori that (some of) the predictors are informative for the response variable and that therefore $\sigma^2$, the amount of variance of $y$ unexplained by the predictors, must be small.
 It therefore makes sense to seek an analytic approximation to the partition function for large values of $\tau$. However, the usual approach to approximate $e^{-\tau H(x)}$ by a Gaussian in the vicinity of the minimizer of $H$ and apply a Laplace approximation \cite{wong2001asymptotic} is not feasible, because $H$ is not twice differentiable. Instead we observe that $e^{-\tau H(x)} = e^{-2\tau f(x)} e^{-2\tau g(x)}$ where
\begin{align}
  f(x) &= \frac12 x^TCx- w^Tx \label{eq:f}\\
  g(x) &= \mu\sum_{j=1}^p |x_j|. \label{eq:g}
\end{align}
Using Parseval's identity for Fourier transforms (Appendix \ref{sec:four-transf-conv}), it follows that (Appendix \ref{sec:four-transf-deta})
\begin{equation}\label{eq:Z}
  Z = \int_{\R^p} e^{-2\tau f(x)} e^{-2\tau g(x)} dx \\
  =  \frac{1}{(\pi\tau)^{\frac{p}2}\sqrt{\det(C)}} \int_{\R^p} e^{-\tau (k-iw)^T C^{-1} (k-iw)}  \prod_{j=1}^p \frac{\mu}{k_j^2+\mu^2} dk. 
\end{equation}
After a change of variables $z=- ik$, $Z$ can be written as a $p$-dimensional complex contour integral
\begin{equation}\label{eq:6}
  Z = \frac{(-i\mu)^p}{(\pi\tau)^{\frac{p}2}\sqrt{\det(C)}} 
      \int_{-i\infty}^{i\infty}\dots \int_{-i\infty}^{i\infty} e^{\tau(z-w)^TC^{-1} (z-w)}
      \prod_{j=1}^p \frac{1}{\mu^2-z_j^2}\, dz_1\dots dz_p.
    \end{equation}
    Cauchy's theorem \cite{lang2002,schneidemann2005} states that this integral remains invariant if the integration contours are deformed, as long as we remain in a domain where the integrand does not diverge (Appendix \ref{sec:cauchys-theor-coord}). The analogue of Laplace's approximation for complex contour integrals, known as the stationary phase, steepest descent or saddle point approximation, then states that an integral of the form \eqref{eq:6} can be approximated by a Gaussian integral along a steepest descent contour passing through the saddle point of the argument of the exponential function \cite{wong2001asymptotic}. Here, the function $(z-w)^TC^{-1}(z-w)$ has a saddle point at $z=w$.  If $|w_j|<\mu$ for all $j$,  the standard stationary phase approximation can be applied directly, but this only covers the uninteresting situation where the maximum-likelihood solution $\xh=\argmin_x H(x)=0$ (Appendix \ref{sec:stat-phase-appr-zero}).   As soon as $|w_j|>\mu$ for at least one $j$, the standard argument breaks down, since to deform the integration contours from the imaginary axes to parallel contours passing through the saddle point $z_0=w$, we would have to pass through a pole (divergence) of the function $\prod_j (\mu^2-z_j^2)^{-1}$ (Figure\ \ref{fig:theory}). Motivated by similar, albeit one-dimensional, analyses in non-equilibrium physics \cite{van2004extended,lee2013modified}, we instead consider a temperature-dependent function
    \begin{equation}\label{eq:htau}
  \Htau(z) = (z-w)^TC^{-1} (z-w) -\frac1\tau\sum_{j=1}^p  \ln(\mu^2-z_j^2),
\end{equation}
which is well-defined on the domain $\D=\{z\in\C^p\colon |\Re z_j|<\mu,\; j=1,\dots,p\}$, where $\Re$ denotes the real part of a complex number. This function has a unique saddle point in $\D$, regardless whether $|w_j|<\mu$ or not (Figure~\ref{fig:theory}). Our main result is a steepest descent approximation of the partition function around this saddle point.

\begin{thm}\label{thm:main}
  Let $C\in\R^{p\times p}$ be a positive definite matrix, $w\in\R^p$ and $\mu>0$. Then the complex function $\Htau$ defined in eq.~\eqref{eq:htau}  has a unique saddle point $\ut$ that is real, $\ut\in \D\cap \R^p$, and is a solution of the set of third order equations
  \begin{equation}\label{eq:03}
    (\mu^2-u_j^2)[C^{-1}(w-u)]_j - \frac{u_j}{\tau} = 0\;,\quad u\in\R^p,\; j\in\{1,\dots,p\}.
  \end{equation}
  For $Q(z)$ a complex analytic function of $z\in\C^p$, the generalized partition function
  \begin{align*}
    Z[Q] =  \frac{1}{(\pi\tau)^{\frac{p}2}\sqrt{\det(C)}} \int_{\R^p} e^{-\tau (k-iw)^T C^{-1} (k-iw)} Q(-ik) \prod_{j=1}^p \frac{\mu }{k_j^2+\mu^2} dk. 
  \end{align*}
  can be analytically expressed as
  \begin{multline}\label{eq:Z-anal}
    Z[Q] = \Bigl(\frac{\mu}{\sqrt\tau}\Bigr)^p e^{\tau (w-\ut )^T C^{-1}
      (w-\ut)} \prod_{j=1}^p\frac{1}{\sqrt{\mu^2+\utj^2}}
    \frac{1}{\sqrt{\det(C+D_\tau)}} \\
    \exp\Bigl\{\frac1{4\tau^2}\Delta_\tau\Bigr\} e^{ R_\tau(ik)} Q(\ut+ik)\biggr|_{k=0},
  \end{multline}
  where
  $D_\tau$ is a diagonal matrix with diagonal elements
  \begin{equation}\label{eq:Dtau}
    (D_\tau)_{jj} = \frac{\tau(\mu^2 - \utj^2)^2}{\mu^2+\utj^2}, 
  \end{equation}
  $\Delta_\tau$ is the differential operator
  \begin{equation}\label{eq:Delta-tau}
      \Delta_\tau =  \sum_{i,j=1}^p \bigl[\tau D_\tau(C+D_\tau)^{-1}C\bigr]_{ij} \frac{\partial^2}{\partial k_i\partial k_j} 
    \end{equation}
    and
  \begin{equation}\label{eq:Rtau}
    R_\tau(z) = \sum_{j=1}^p \sum_{m\geq 3} \frac1m
    \Bigl[\frac{1}{(\mu - \utj)^m} + \frac{(-1)^m}{(\mu +
      \utj)^m} \Bigr] z_j^m.
  \end{equation}
  This results in an analytic approximation
  \begin{equation}\label{eq:Z-approx}
    Z[Q] \sim \Bigl(\frac{\mu}{\sqrt\tau}\Bigr)^p e^{\tau (w-\ut )^T C^{-1}
      (w-\ut)} \prod_{j=1}^p\frac{1}{\sqrt{\mu^2+\utj^2}}
    \frac{Q(\ut)}{\sqrt{\det(C+D_\tau)}} .
  \end{equation}
\end{thm}

The analytic expression in eq.~\eqref{eq:Z-anal} follows by changing the integration contours to pass through the saddle point $\ut$, and using a Taylor expansion of $\Htau(z)$ around the saddle point along the steepest descent contour. However, because $\Delta_\tau$ and $R_\tau$ depend on $\tau$, it is not a priori evident that \eqref{eq:Z-approx} holds. A detailed proof is given in Appendix \ref{sec:proof-theorem}. The analytic approximation in eq.~\eqref{eq:Z-approx} can be simplified further by expanding $\ut$ around its leading term, resulting in an expression that recognizably converges to the sparse maximum-likelihood solution (Appendix \ref{sec:zero-temp-limit}). While eq.~\eqref{eq:Z-approx} is computationally more expensive to calculate than the corresponding expression in terms of the maximum-likelihood solution, it was found to be numerically more accurate (Section \ref{sec:numer-exper}). 

Various quantities derived from the posterior distribution can be expressed in terms of generalized partition functions. The most important of these are the expectation values of the regression coefficients, which, using elementary properties of the Fourier transform  (Appendix \ref{sec:gener-part-funct}), can be expressed as
\begin{align*}
  \E(x) &= \frac1{Z}\int_{\R^p} x\, e^{-\tau H(x)} = \frac{Z\bigl[C^{-1}(w-z)\bigr]}{Z}
  \sim C^{-1}(w -\ut).
\end{align*}
The leading term,
\begin{equation}
  \label{eq:xbay}
  \xt\equiv C^{-1}(w -\ut),
\end{equation}
can be interpreted as an estimator for the regression coefficients in its own right, which interpolates smoothly (as a function of $\tau$) between the ridge regression estimator $\xr = C^{-1}w$ at $\tau=0$ and the maximum-likelihood elastic net estimator $\xh = C^{-1}(w-\hat u)$ at $\tau=\infty$, where $\hat u = \lim_{\tau\to\infty}\ut$ satisfies a box-constrained optimization problem (Appendix \ref{sec:zero-temp-limit}).


The marginal posterior distribution for a subset $I\subset\{1,\dots,p\}$ of regression coefficients is defined as
\begin{align*}
p(x_I) &= \frac1{Z(C,w,\mu)} \int_{\R^{|I^c|}} e^{-\tau H(x\mid C,w,\mu)} dx_{I^c}
\end{align*}
where $I^c = \{1,\dots,p\}\setminus I$ is the complement of $I$, $|I|$ denotes the size of a set $I$, and we have reintroduced temporarily the dependency on $C$, $w$ and $\mu$ as in eq.~\eqref{eq:2}. A simple calculation shows that the remaining integral is again a partition function of the same form, more precisely:
\begin{equation}\label{eq:marginal}
p(x_I)= e^{-\tau(x_I^TC_{I}x_I - 2w_I^Tx_I + 2\mu \|x_I\|_1)} \frac{Z(C_{I^c},w_{I^c}-x_I^TC_{I,I^c},\mu)}{Z(C,w,\mu)},
\end{equation}
where the subscripts $I$ and $I^c$ indicate sub-vectors and sub-matrices on their respective coordinate sets. Hence the analytic approximation in eq.~\eqref{eq:Z-anal} can be used to approximate numerically each term in the partition function ratio and obtain an approximation to the marginal posterior distributions. 

The posterior predictive distribution \cite{friedman2001elements} for a new sample $a\in\R^p$ of predictor data can also be written as a ratio of partition functions:
\begin{align*}
  p(y) = \int_{\R^p} p(y\mid a,x) p(x\mid C,w,\mu)\,dx = \Bigl(\frac{\tau}{2\pi n}\Bigr)^{\frac 12} e^{-\frac\tau{2n}y^2} \frac{Z\bigl(C+\frac1{2n}aa^T,w+\frac{y}{2n}a,\mu\bigr)}{Z(C,w,\mu)},
\end{align*}
where $C\in\R^{p\times p}$ and $w\in\R^p$ are obtained from the training data as before, $n$ is the number of training samples, and $y\in\R$ is the unknown response to $a$ with distribution $p(y)$. Note that
\begin{multline*}
  \E(y) = \int_{\R} y p(y) dy = \int_{\R^p} \Bigl[\int_\R p(y\mid a,x) dy\Bigr] p(x\mid C,w,\mu)\,dx\\
  = \int_{\R^p} a^Tx p(x\mid C,w,\mu)\,dx = a^T \E(x) \sim a^T \xt.
\end{multline*}

\section{Numerical experiments}
\label{sec:numer-exper}

To test the accuracy of the stationary phase approximation, we implemented algorithms to solve the saddle point equations and compute the partition function and marginal posterior distribution, as well as an existing Gibbs sampler algorithm \cite{hans2011elastic} in Matlab (see Appendix \ref{sec:numerical-recipes} for algorithm details, source code available from \url{https://github.com/tmichoel/bayonet/}). Results were first evaluated for independent predictors (or equivalently, one predictor) and two commonly  used data sets: the ``diabetes data'' ($n=442$, $p=10$) \cite{efron2004least} and the ``leukemia data'' ($n=72$, $p=3571$) \cite{zou2005regularization} (see Appendix \ref{sec:experimental-details} for further experimental details and data sources).

\begin{figure}
  \centering
  \includegraphics[width=\linewidth]{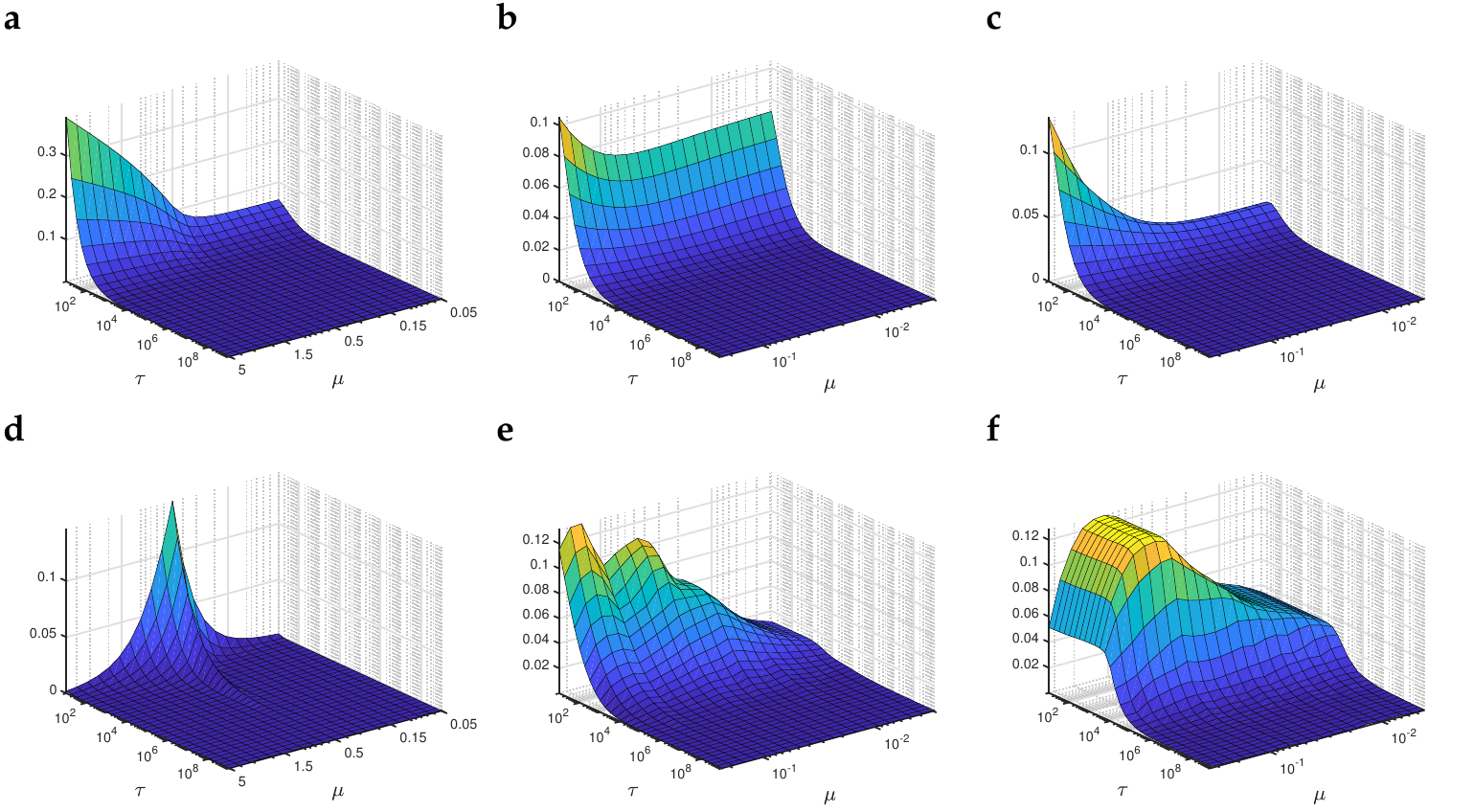}
  \caption{Convergence to the minimum-energy solution. Top row: $(-\frac1\tau\log Z-H_{\min})/p$ vs.\ $\tau$ and $\mu$ for the exact partition function for independent predictors ($p=1$) (\textbf{a}), and for the stationary phase approximation for the diabetes (\textbf{b}) and leukemia (\textbf{c}) data. Bottom row: $\|\xt-\xh\|_\infty$ for the exact expectation value for independent predictors (\textbf{d}), and using the stationary phase approximation for the diabetes (\textbf{e}) and leukemia (\textbf{f}) data. Parameter values were $C=1.0$, $w=0.5$, and $\mu$ ranging from $0.05$ to $5$ in geometric steps (\textbf{a}), and $\lambda=0.1$ and $\mu$ ranging from $0.01 \mu_{\max}$ upto, but not including, $\mu_{\max}=\max_j |w_j|$ in geometric steps (\textbf{b},\textbf{c}).}
  \label{fig:logZ-convergence}
\end{figure}

First we tested the rate of convergence in the asymptotic relation (see Appendix \ref{sec:zero-temp-limit})
\begin{align*}
  \lim_{\tau\to\infty} -\frac1\tau \log Z = H_{\min} = \min_{x\in\R^p} H(x).
\end{align*}
For independent predictors ($p=1$), the partition function can be calculated analytically using the error function (Appendix \ref{sec:analyt-results-indep}), and rapid convergence to $H_{\min}$ is observed (Figure\ \ref{fig:logZ-convergence}a). After scaling by the number of predictors $p$, a similar rate of convergence is observed for the stationary phase approximation to the partition function for both the diabetes and leukemia data (Figure\ \ref{fig:logZ-convergence}b,c). However, convergence of the posterior expectation values $\xt$ to the maximum-likelihood coefficients $\xh$, as measured by the $\ell_\infty$-norm difference $\|\xt-\xh\|_\infty=\max_j|\xtj-\xhj|$ is noticeably slower, particularly in the $p\gg n$ setting of the leukemia data (Figure\ \ref{fig:logZ-convergence}d--f).

\begin{figure}
  \centering
  \includegraphics[width=\linewidth]{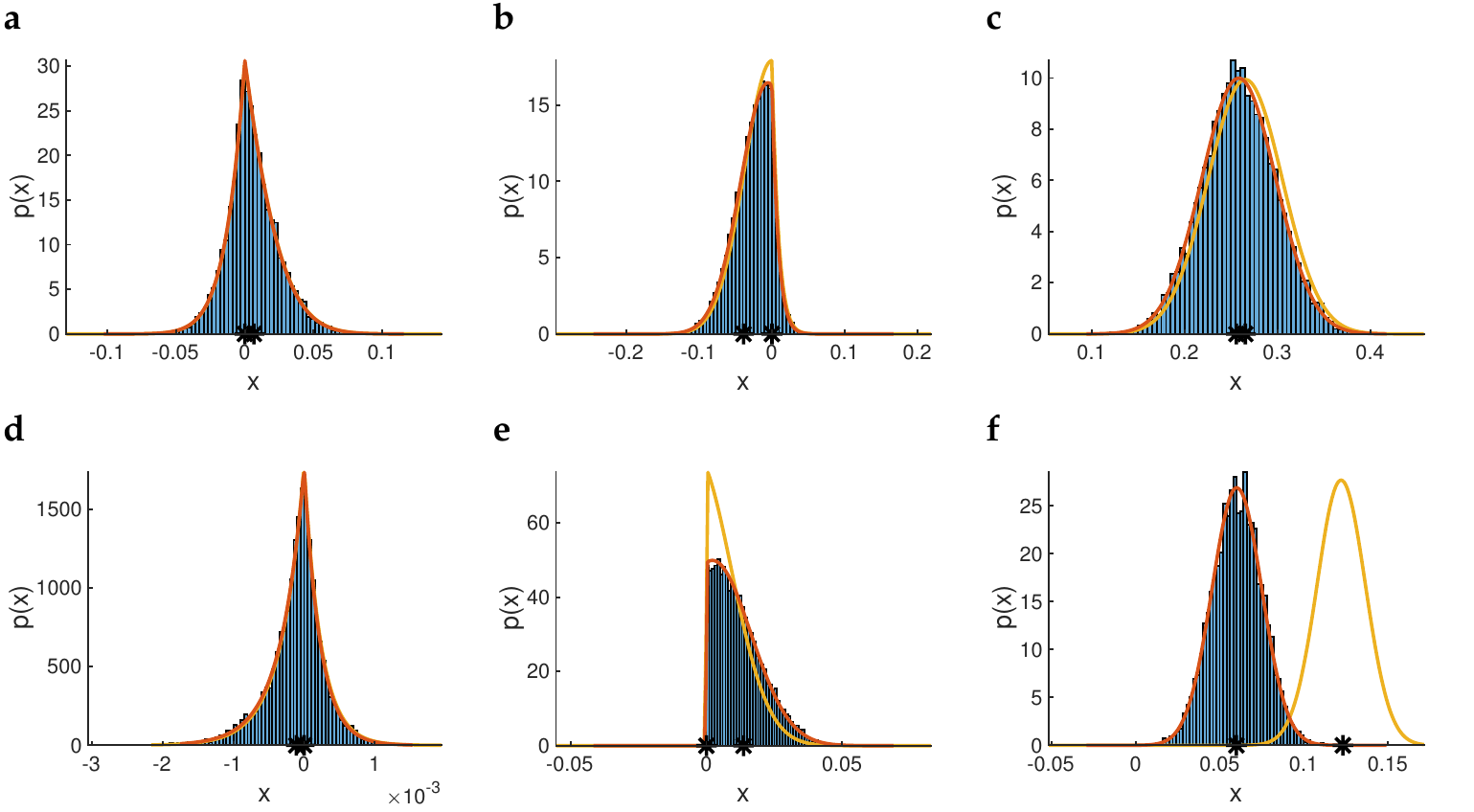}
  \caption{Marginal posterior distributions for the diabetes data ($\lambda=0.1$, $\mu=0.0397$, $\tau=682.3$) (\textbf{a--c}) and leukemia data ($\lambda=0.1$, $\mu=0.1835$, $\tau=9943.9$) (\textbf{d--f};).  In blue, Gibbs sampling histogram ($10^4$ samples). In red, stationary phase approximation for the marginal posterior distribution of selected predictors. In yellow, maximum-likelihood-based approximation for the same distributions. The distributions for a zero, transition and non-zero maximum-likelihood predictor are shown (from left to right). The {\large $\ast$} on the $x$-axes indicate the location of the maximum-likelihood and posterior expectation value.}
  \label{fig:marginal-dist}
\end{figure}

Next, the accuracy of the stationary phase approximation at finite $\tau$ was determined by comparing the marginal distributions for single predictors [i.e. where $I$ is a singleton in eq.~\eqref{eq:marginal}] to results obtained from Gibbs sampling. For simplicity, representative results are shown for specific hyper-parameter values (Appendix \ref{sec:diab-leuk-data}). Application of the stationary phase approximation resulted in marginal posterior distributions which were indistinguishable from those obtained by Gibbs sampling (Figure~\ref{fig:marginal-dist}). In view of the convergence of the log-partition function to the minimum-energy value (Figure~\ref{fig:logZ-convergence}), an approximation to eq.~\eqref{eq:marginal} of the form
\begin{equation}\label{eq:22}
  p(x_I)\approx e^{-\tau(x_I^TC_{I}x_I - 2w_I^Tx_I + 2\mu \|x_I\|_1)}  e^{-\tau[H_{\min}(C_{I^c},w_{I^c}-x_I^TC_{I,I^c},\mu) - H_{\min}(C,w,\mu)]}
\end{equation}
was also tested. However, while eq.~\eqref{eq:22} is indistinguishable from eq.~\eqref{eq:marginal} for predictors with zero effect size in the maximum-likelihood solution, it resulted in distributions that were squeezed towards zero for transition predictors, and often wildly inaccurate for non-zero predictors (Figure~\ref{fig:marginal-dist}). This is because eq.~\eqref{eq:22} is maximized at $x_I=\hat x_I$, the maximum-likelihood value, whereas for non-zero coordinates, eq.~\eqref{eq:marginal} is (approximately) symmetric around its expectation value $\E(x_I)=\hat x_{\tau,I}$. Hence, accurate estimations of the marginal posterior distributions requires using the full stationary phase approximations [eq.~\eqref{eq:Z-approx}] to the partition functions in eq. \eqref{eq:marginal}.

\begin{figure}
  \centering
  \includegraphics[width=\linewidth]{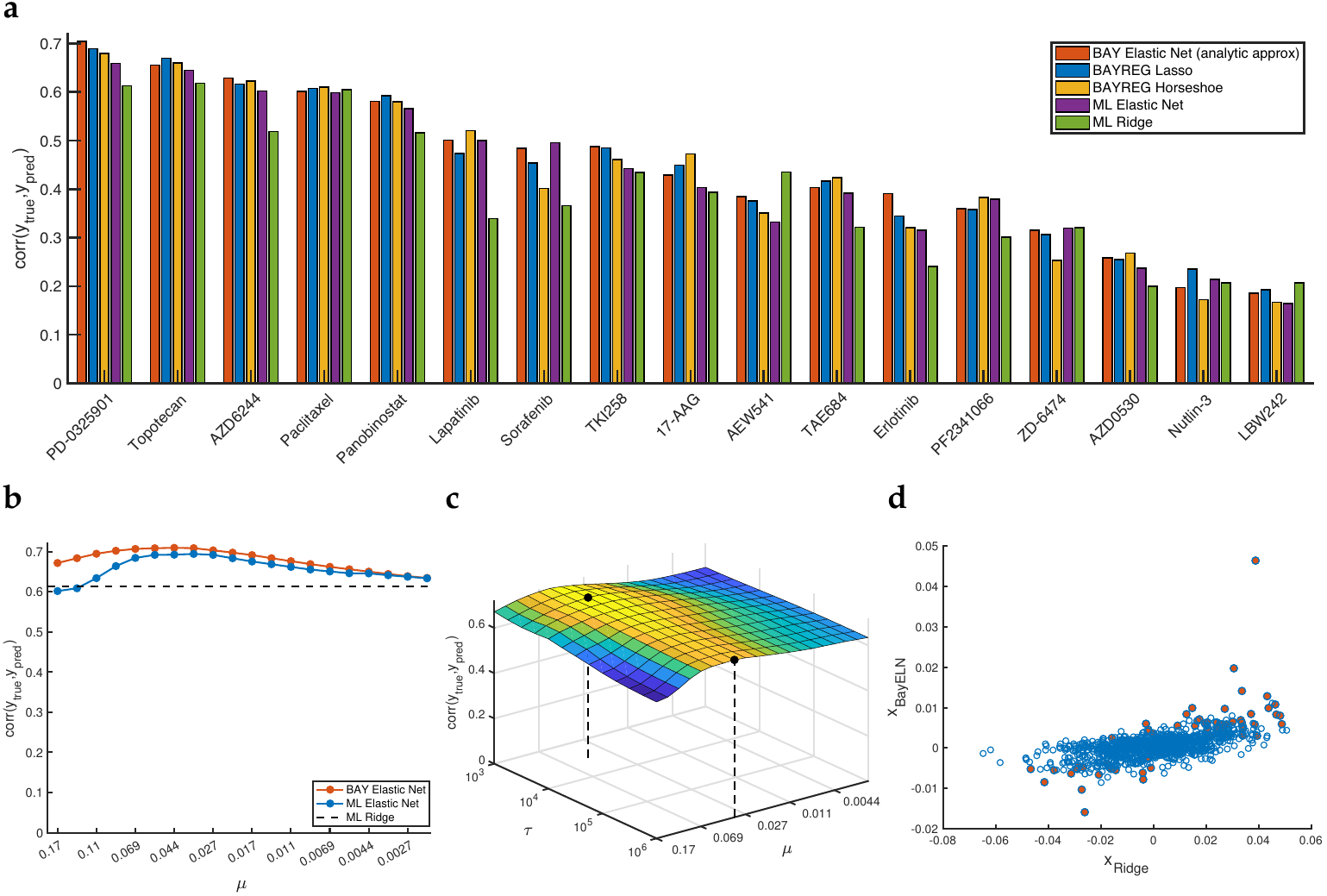}
  \caption{Predictive accuracy on the Cancer Cell Line Encyclopedia. \textbf{a.} Median correlation coefficient between predicted and true drug sensitivities over 10-fold cross-validation, using Bayesian posterior expectation values from the analytic approximation for elastic net (red) and from BayReg's lasso (blue) and horseshoe (yellow) implementations, and maximum-likelihood elastic net (purple) and ridge regression (green) values for the regression coefficients. See main text for details on hyper-parameter optimization. \textbf{b.} Median 10-fold cross-validation value for the correlation coefficient between predicted and true sensitivities for the compound PD-0325901 vs.\ $\mu$, for the Bayesian elastic net at optimal $\tau$ (red), maximum-likelihood elastic net (blue) and ridge regression (dashed). \textbf{c.} Median 10-fold cross-validation value for the correlation coefficient between predicted and true sensitivities for PD-0325901 for the Bayesian elastic net vs.\ $\tau$ and $\mu$; the black dots show the overall maximum and the ML maximum. \textbf{d.} Scatter plot of expected regression coefficients in the Bayesian elastic net for PD-0325901 at $\mu=0.055$ and optimal $\tau=3.16\cdot 10^3$ vs.\ ridge regression coefficient estimates; coefficients with non-zero maximum-likelihood elastic net value at the same $\mu$ are indicated in red. See Supp.\ Figures \ref{fig:ccle-sup-1} and \ref{fig:ccle-sup-2} for the other 16 compounds.}
  \label{fig:ccle}
\end{figure}

The stationary phase approximation can be particularly advantageous in prediction problems, where the response value $\hat y\in\R$ for a newly measured predictor sample $a\in\R^{p}$ is obtained using regression coefficients learned from training data $(y_t,A_t)$. In Bayesian inference, $\hat y$ is set to the expectation value of the posterior predictive distribution, $\hat y = \E(y) = a^T\xt$ [eq.~\eqref{eq:4}]. Computation of the posterior expectation values $\xt$ [eq.~\eqref{eq:xbay}] using the stationary phase approximation requires solving only one set of saddle point equations, and hence can be performed efficiently across a range of hyper-parameter values, in contrast to Gibbs sampling, where the full posterior needs to be sampled even if only expectation values are needed. 

To illustrate how this benefits large-scale applications of the Bayesian elastic net, its prediction performance was compared to state-of-the-art Gibbs sampling implementations of Bayesian horseshoe and Bayesian lasso regression \cite{makalic2016high}, as well as to maximum-likelihood elastic net and ridge regression, using gene expression and drug sensitivity data for 17 anticancer drugs in 474 human cancer cell lines from the Cancer Cell Line Encyclopedia \cite{barretina2012cancer}. Ten-fold cross-validation was used, using $p=1000$ pre-selected genes and 427 samples for training regression coefficients and 47 for validating predictions in each fold. To obtain unbiased predictions at a single choice for the hyper-parameters, $\mu$ and $\tau$ were optimized over a $10\times 13$ grid using an additional internal 10-fold cross-validation on the training data only (385 samples for training, 42 for testing); BayReg's lasso and horseshoe methods sample hyper-parameter values from their posteriors and do not require an additional cross-validation loop (see Appendix \ref{sec:cencer-cell-line} for complete experimental details and data sources). Despite evaluating a much greater number of models (in each cross-validation fold, 10$\times$ cross-validation over 130 hyper-parameter combinations vs.\ 1 model per fold), the overall computation time was still much lower than BayReg's Gibbs sampling approach (on average 30 sec. per fold, i.e.\ 0.023 sec.\ per model, vs. 44 sec. per fold for BayReg). In terms of predictive performance, Bayesian methods tended to perform better than maximum-likelihood methods, in particular for the most `predictable' responses, with little variation between the three Bayesian methods (Figure \ref{fig:ccle}a). 

While the difference in optimal performance between Bayesian and maximum-likelihood elastic net was not always large, Bayesian elastic net tended to be optimized at larger values of $\mu$ (i.e. at sparser maximum-likelihood solutions), and at these values the performance improvement over maximum-likelihood elastic net was more pronounced (Figure~\ref{fig:ccle}b). As expected, $\tau$ acts as a tuning parameter that allows to smoothly vary from the maximum-likelihood solution at large $\tau$ (here, $\tau\sim10^6$) to the solution with best cross-validation performance (here, $\tau\sim10^3-10^4$) (Figure~\ref{fig:ccle}c). The improved performance at sparsity-inducing values of $\mu$ suggests that the Bayesian elastic net is uniquely able to identify the dominant predictors for a given response (the non-zero maximum-likelihood coefficients), while still accounting for the cumulative contribution of predictors with small effects. Comparison with the unpenalized ($\mu=0$) ridge regression coefficients shows that the Bayesian expectation values are strongly shrunk towards zero, except for the non-zero maximum-likelihood coefficients, which remain relatively unchanged (Figure~\ref{fig:ccle}d), resulting in a double-exponential distribution for the regression coefficients. This contrasts with ridge regression, where regression coefficients are normally distributed leading to over-estimation of small effects, and maximum-likelihood elastic net, where small effects become identically zero and don't contribute to the predicted value at all.

\section{Conclusions}
\label{sec:conclusions}

The application of Bayesian methods to infer expected effect sizes and marginal posterior distributions in $\ell_1$-penalized models has so far required the use of computationally expensive Gibbs sampling methods. Here it was shown that highly accurate inference in these models is also possible using an analytic stationary phase approximation to the partition function integrals. This approximation exploits the fact that the Fourier transform of the non-differentiable double-exponential prior distribution is a well-behaved exponential of a log-barrier function, which is intimately related to the Legendre-Fenchel transform of the $\ell_1$-penalty term. Thus, the Fourier transform plays the same role for Bayesian inference problems as convex duality plays for maximum-likelihood approaches. 

For simplicity, we have focused on the linear regression model, where the invariance of multivariate normal distributions under the Fourier transform greatly facilitates the analytic derivations.  Preliminary work shows that the results can probably be extended to generalized linear models (or any model with convex energy function) with L1 penalty, using the argument sketched in Appendix \ref{sec:conv-fourier}. In such models, the predictor correlation matrix $C$ will need to be replaced by the Hessian matrix of the energy function evaluated at the saddle point. Numerically, this will require updates of the Hessian during the coordinate descent algorithm for solving the saddle point equations. How to balance the accuracy of the approximation and the frequency of the Hessian updates will require further in-depth investigation.  In principle, the same analysis can also be performed using other non-twice-differentiable sparse penalty functions, but if their Fourier transform is not known analytically, or not twice differentiable either, the analysis and implementation will become more complicated still.
 
A limitation of the current approach may be that values of the hyper-parameters need to be specified in advance, whereas in complete hierarchical models, these are subject to their own prior distributions. Incorporation of such priors will require careful attention to the interchange between taking the limit of and integrating over the inverse temperature parameter. However, in many practical situations  $\ell_1$ and $\ell_2$-penalty parameters are pre-determined by cross-validation. Setting the residual variance parameter to its maximum a-posteriori value then allows to evaluate the maximum-likelihood solution in the context of the posterior distribution of which it is the mode \cite{hans2011elastic}. Alternatively, if the posterior expectation values of the regression coefficients are used instead of their maximum-likelihood values to predict unmeasured responses, the optimal inverse-temperature parameter can be determined by standard cross-validation on the training data, as in the drug response prediction experiments.

No attempt was made to optimize the efficiency of the coordinate descent algorithm  to solve the saddle point equations. However, comparison to the Gibbs sampling algorithm shows that one cycle through all coordinates in the coordinate descent algorithm is approximately equivalent to one cycle in the Gibbs sampler, i.e.\ to adding one more sample. The coordinate descent algorithm typically converges in 5-10 cycles starting from the maximum-likelihood solution, and 1-2 cycles when starting from a neighbouring solution in the estimation of marginal distributions. In contrast, Gibbs sampling typically requires $10^3$-$10^5$ coordinate cycles to obtain stable distributions. Hence, if only the posterior expectation values or the posterior distributions for a limited number of coordinates are sought, the computational advantage of the stationary phase approximation is vast. On the other hand, each evaluation of the marginal distribution functions requires the solution of a separate set of saddle point equations. Hence, computing these distributions for all predictors at a very large number of points with the current algorithm could become equally expensive as Gibbs sampling. 

In summary, expressing intractable partition function integrals as complex-valued oscillatory integrals through the Fourier transform is a powerful approach for performing Bayesian inference in the lasso and elastic net regression models, and $\ell_1$-penalized models more generally. Use of the stationary phase approximation to these integrals results in highly accurate estimates for the posterior expectation values and marginal distributions at a much reduced computational cost compared to Gibbs sampling.

\subsubsection*{Acknowledgments}


This research was supported by the BBSRC (grant numbers BB/J004235/1 and BB/M020053/1).



\clearpage{}

\appendix

\begin{center}
  {\Large \textbf{Appendices}}
\end{center}

\section{Basic results in Fourier space}
\label{sec:part-funct-four}

\subsection{Fourier transform conventions}
\label{sec:four-transf-conv}

Fourier transforms are defined with different scaling conventions in different branches of science. Here, the symmetric version of the Fourier transform written in terms of angular frequencies is used: for $f$ a function on $\R^p$, we define
\begin{align*}
  \F[f](k) &= \hat f(k) = \frac1{(2\pi)^{\frac{p}2}}\int_{\R^p} f(x) e^{-i k^Tx} dx
\end{align*}
and
\begin{align*}
  f(x) = \F^{-1}\bigl[\F[f]\bigr] (x) = \frac1{(2\pi)^{\frac{p}2}}\int_{\R^p} \hat f(k) e^{i k^Tx}  dk.
\end{align*}
Parseval's identity states that for two functions $f$ and $g$,
\begin{align*}
  \int_{\R^p} \overline{f(x)} g(x) dx = \int_{\R^p}  \overline{\hat f(k)} \hat g(k) dk,
\end{align*}
where $\bar{\cdot}$ denotes complex conjugation. For more details, see \cite[Chapter 11]{hunter2001applied}.

\subsection{Relation between convex duality and the Fourier transform}
\label{sec:conv-fourier}

The motivation for using the Fourier transform to study Bayesian inference problems stems from the correspondence between the Fourier and Legendre-Fenchel transforms of convex functions. This correspondence is an example of so-called idempotent mathematics, and a survey of its history and applications can be found in \cite{litvinov2005maslov}, while a formal treatment along the lines below can be found in \cite{fedoryuk1971stationary}, and a summary of analogous properties between the Legendre-Fenchel and Fourier transforms can be found in \cite{alonso1995fractional}. The basic argument is presented here, without any attempt at being complete or rigorous.

Let $h$ be a convex function on $\R^p$ and assume it is sufficiently smooth for the statements below to hold without needing too much attention to detail. The Gibbs probability distribution for $h$ at inverse temperature $\tau$ is defined as $p(x)=\frac1Z e^{-\tau h(x)}$, with $Z=\int_{\R^p} e^{-\tau h(x)}dx$ the partition function. Define for $z\in\C^p$
\begin{align*}
  h^\ast_\tau (z) = \frac1\tau \ln \int_{\R^p} e^{-\tau[h(x)-z^Tx]}dx.
\end{align*}
By the Laplace approximation, it follows that for $\tau$ large and $u\in\R^p$, to leading order in $\tau$,
\begin{equation}\label{eq:36}
  h^\ast_\tau (u)\approx h^*(u) = \max_{x\in\R^p} [u^Tx-h(x)],
\end{equation}
the Legendre-Fenchel transform of $h$. The Fourier transform of $e^{-\tau h}$ is
\begin{equation}\label{eq:30}
  \F\bigl[e^{-\tau h}\bigr] (\tau k) = \frac1{(2\pi)^{\frac{p}2}}\int_{\R^p} e^{-\tau h(x)} e^{-i\tau k^Tx} dx =  \frac{e^{\tau h^\ast_\tau(-ik)}}{(2\pi)^{\frac{p}2}}.
\end{equation}
Now assume that $h=f+g$ can be written as the sum of two convex functions $f$ and $g$. It is instructive to think of $h(x)$ as minus a posterior log-likelihood function of regression coefficients $x$, with a natural decomposition in a part $f(x)$ coming from the data likelihood and a part $g(x)$ representing the prior distribution on $x$. We again assume that $f$ and $g$ are smooth.

The Parseval identity for Fourier transforms yields
\begin{align*}
  \int_{\R^p} e^{-\tau[f(x)+g(x)]} dx &= \int_{\R^p} \overline{\F\bigl[e^{-\tau f}\bigr] (k)} \F\bigl[e^{-\tau g}\bigr] (k) dk
  = \Bigl(\frac{\tau}{2\pi}\Bigr)^p \int_{\R^p} e^{\tau [f^\ast_\tau(ik) + g^\ast_\tau(-ik)]} dk,
\end{align*}
where a change of variables $k\to \tau k$ was made. When $\tau$ is large, the Laplace approximation of the l.h.s. states that, to leading order in $\tau$
\begin{equation}\label{eq:1app}
  \frac1\tau\ln \int_{\R^p} e^{-\tau[f(x)+g(x)]} dx \approx  - \min_{x\in\R^p} \bigl[f(x)+g(x)\bigr] = \max_{x\in\R^p} \bigl[-f(x)-g(x)\bigr].
\end{equation}
The integral on the r.h.s. can be written as a complex contour integral
\begin{align*}
  \int_{\R^p} e^{\tau [f^\ast_\tau(ik) + g^\ast_\tau(-ik)]} dk = \frac1{i^p} \int_{i\R^p}  e^{\tau [f^\ast_\tau(z) + g^\ast_\tau(-z)]} dz, 
\end{align*}
where $i\R^p$ denotes a $p$-dimensional contour consisting of vertical contours running along the imaginary axis in each dimension. The steepest descent or saddle point approximation \citep{wong2001asymptotic} requires that we deform the contour to run through the saddle point, i.e. a zero of the gradient function $\nabla  [f^\ast_\tau(z) + g^\ast_\tau(-z)]$. Under fairly general conditions (see for instance \cite{daniels1954saddlepoint}), $f^\ast_\tau(z) + g^\ast_\tau(-z)$ will attain its maximum modulus at a real vector, and hence the new integration contour will take the form $z=\ut+ik$ where $\ut=\argmin_{u\in\R^p} [f^\ast_\tau(u) + g^\ast_\tau(-u)]$ and $k\in\R^p$. Note that in the limit $\tau\to\infty$, $\ut\to \uh=\argmin_{u\in\R^p} [f^\ast(u) + g^\ast(-u)]$. The stationary phase approximation yields, again to leading order in $\tau$
\begin{multline}\label{eq:2app}
  \frac1\tau\ln\int_{\R^p} e^{\tau [f^\ast_\tau(ik) + g^\ast_\tau(-ik)]} dk 
  =  \frac1\tau\ln\int_{\R^p} e^{\tau [f^\ast_\tau(\ut+ik) + g^\ast_\tau(-\ut-ik)]} dk \\
  \approx \min_{u\in\R^p} \bigl[f^\ast_\tau(u) + g^\ast_\tau(-u)\bigr] 
  \approx \min_{u\in\R^p} \bigl[f^\ast(u) + g^\ast(-u)\bigr] 
\end{multline}
Combining eqs. \eqref{eq:1app} and \eqref{eq:2app}, we recover Fenchel's well-known duality theorem
\begin{align*}
  \max_{x\in\R^p} \bigl[-f(x)-g(x)\bigr] = \min_{u\in\R^p} \bigl[f^\ast(u) + g^\ast(-u)\bigr] .
\end{align*}

In summary, there is an equivalence between convex duality for log-likelihood functions and switching from coordinate to frequency space using the Fourier transform for Gibbs probability distributions, which becomes an exact mapping in the limit of large inverse temperature. As shown in this paper, this remains true even when $f$ or $g$ are not necessarily smooth (e.g.\ if $g(x)=\|x\|_1$ is the $\ell_1$-norm). 

\subsection{The Fourier transform of the multivariate normal and Laplace distributions}
\label{sec:four-transf-deta}

To derive eq. \eqref{eq:Z}, observe that $f(x)$ is a Gaussian and its Fourier transform is again a Gaussian: 
\begin{multline}\label{eq:37}
  \overline{\F(e^{-2\tau f})} = \frac1{(2\pi)^{\frac{p}2}}\int_{\R^p}
  e^{-2\tau f(x)} e^{ik^Tx} dx\\
  = \frac{1}{\sqrt{(2\tau)^p\det(C)}}
  \exp\left\{-\frac{1}{4\tau} (k -2i\tau w)^T C^{-1} (k -
    2i\tau w)\right\}.
\end{multline}
To calculate the Fourier transform of $e^{-\tau g}$, note that in one dimension
\begin{equation*}
  \int_\R e^{-\gamma|x|} e^{- i kx} dx = \frac{2\gamma}{k^2 +
    \gamma^2}, 
\end{equation*}
and hence
\begin{align*}
  \F(e^{-2\tau g})(k) = \frac1{(2\pi)^{\frac{p}2}}\prod_{j=1}^p \frac{4\mu\tau}{k_j^2 + 4\tau^2\mu^2}.
\end{align*}
After making the change of variables $k_j'=\frac1{2\tau} k_j$, eq. \eqref{eq:Z} is obtained.

\subsection{Cauchy's theorem in coordinate space}
\label{sec:cauchys-theor-coord}

Cauchy's theorem \citep{lang2002,schneidemann2005} states that we can freely deform the integration contours in the integral in eq. \eqref{eq:6} as long as we remain within a holomorphic domain of the integrand, or simply put, a domain where the integrand does not diverge. Consider as a simple example the deformation of the integration contours from $z_j\in i\R$ in eq. \eqref{eq:6} to $z_j\in w_j'+i\R$, where $|w'_j|<\mu$ for all $j$. We obtain
\begin{align*}
  Z &= \frac{(-i\mu)^p}{(\pi\tau)^{\frac{p}2}\sqrt{\det(C)}} 
      \int_{w'_1-i\infty}^{w'_1+i\infty}\dots \int_{w'_p-i\infty}^{w'_p+i\infty} e^{\tau (z-w)^TC^{-1} (z-w)}
      \prod_{j=1}^p \frac{1}{\mu^2-z_j^2}\, dz_1\dots dz_p\\
  &= \frac{\mu^p}{(\pi\tau)^{\frac{p}2}\sqrt{\det(C)}} \int_{\R^p} e^{-\tau (w'-w+ik)^TC^{-1} (w'-w+ik)}
      \prod_{j=1}^p \frac{1}{\mu^2-(w'_j+ik_j)^2}\, dk, 
\end{align*}
where we parameterized $z_j=w'_j+i k_j$. Using the inverse Fourier transform, and reversing the results from Section \ref{sec:analytic-results} and Appendix \ref{sec:four-transf-deta}, we can write this expression as
\begin{align*}
  Z = \int_{\R^p} e^{-2\tau \tilde f(x)} e^{-2\tau \tilde g(x)},
\end{align*}
where
\begin{align}
  \tilde f (x) &=\frac12 x^TCx-(w-w')^Tx\\
  \tilde g(x) &=\sum_{j=1}^p(\mu|x_j| - w'_j x_j). \label{eq:gpr}
\end{align}
Comparison with eqs. \eqref{eq:f}--\eqref{eq:g} shows that the freedom to deform the integration contour in Fourier space corresponds to an equivalent freedom to split $e^{-\tau H(x)}$ into a product of two functions. Clearly eq. \eqref{eq:gpr} only defines an integrable function $e^{-2\tau\tilde g}$ if $|w'_j|<\mu$ for all $j$, which of course corresponds to the limitation imposed by Cauchy's theorem that the deformation of the integration contours cannot extend beyond the domain where the function $\prod_j (\mu^2-z_j^2)^{-1}$ remains finite.

\subsection{Stationary phase approximation in the zero-effect case}
\label{sec:stat-phase-appr-zero}

Assume that $|w_j|<\mu$ for all $j$. It then follows immediately that the maximum-likelihood or minimum-energy solution $\hat x=\argmin_x H(x) = 0$. As above, we can deform the integration contours in \eqref{eq:6} into steepest descent contours passing through the saddle point $z_0=w$ of the function $h(z)=(z-w)^TC^{-1} (z-w)$ (cf.\ Figure\ \ref{fig:theory}a). We obtain
\begin{align}
  Z &= \frac{(-i\mu)^p}{(\pi\tau)^{\frac{p}2}\sqrt{\det(C)}} 
      \int_{w_1-i\infty}^{w_1+i\infty}\dots \int_{w_p-i\infty}^{w_p+i\infty} e^{\tau (z-w)^TC^{-1} (z-w)}
      \prod_{j=1}^p \frac{1}{\mu^2-z_j^2}\, dz_1\dots dz_p\nonumber\\
  &= \frac{\mu^p}{(\pi\tau)^{\frac{p}2}\sqrt{\det(C)}} \int_{\R^p} e^{-\tau k^TC^{-1}k}
      \prod_{j=1}^p \frac{1}{\mu^2-(w_j+ik_j)^2}\, dk,\label{eq:39}
\end{align}
where we parameterized $z_j=w_j+i k_j$. This integral can be written as a series expansion using the following standard result, included here for completeness.
\begin{lem}\label{lem:gaussian-kernel}
  Let $C\in\R^p\times\R^p$ be a positive definite matrix and let $\Delta_C$ be the differential operator
\begin{align*}
  \Delta_C = \sum_{i,j=1}^p C_{ij} \frac{\partial^2}{\partial k_i \partial k_j}.
\end{align*}
Then
\begin{equation*}
 \frac1{\pi^{\frac{p}2}\sqrt{\det(C)}}  \int_{\R^p} e^{- k^TC^{-1}k} \hat f(k) dk = \Bigl(e^{\frac14\Delta_C}\hat f\Bigr)(0).
\end{equation*}
\end{lem}
\begin{proof}
  First note that
  \begin{equation}\label{eq:40}
    \Delta_C e^{-ik^Tx} = -\sum_{ij}C_{ij} x_i x_j e^{-ik^Tx} =- (x^TCx)\, e^{-ik^Tx} ,
  \end{equation}
  i.e.\ $e^{ik^Tx}$ is an `eigenfunction' of $\Delta_C$ with eigenvalue $- (x^TCx)$, and hence
  \begin{align*}
    e^{\frac14\Delta_C} e^{-ik^Tx} = e^{- \frac14 x^TCx}e^{-ik^Tx}.
  \end{align*}
  Using the (inverse) Fourier transform, we can define
  \begin{align*}
    f(x) &= \frac1{(2\pi)^{\frac{p}2}}\int_{\R^p} \hat f(k) e^{ik^Tx} dk,\\
    \intertext{and write}
    \hat f(k) &= \frac1{(2\pi)^{\frac{p}2}}\int_{\R^p} f(x) e^{-ik^Tx} dx.
  \end{align*}
  Hence
  \begin{align*}
    \Bigl(e^{\frac14\Delta_C} \hat f\Bigr)(k) &= \frac1{(2\pi)^{\frac{p}2}}\int_{\R^p} f(x) e^{\frac14\Delta_C} e^{ik^Tx} dx = \frac1{(2\pi)^{\frac{p}2}}\int_{\R^p} f(x) e^{- \frac14 x^TCx}e^{-ik^Tx}dx.
  \end{align*}
  Using Parseval's identity and the formula for the Fourier transform of a Gaussian [eq.~\eqref{eq:37}], we obtain
  \begin{align*}
    \Bigl(e^{\frac14\Delta_C} \hat f\Bigr)(0) &= \frac1{(2\pi)^{\frac{p}2}}\int_{\R^p} f(x) e^{- \frac14 x^TCx} dx = \frac1{\pi^{\frac{p}2}\sqrt{\det(C)}} \int_{\R^p} \hat f(k) e^{-k^TC^{-1}k}dk
  \end{align*}
\end{proof}
In the derivation above, we have tacitly assumed that the inverse Fourier transform $f$ of $\hat f$ exists. However, the result remains true even if $f$ is only a distribution, i.e. $\hat f$ need not be integrable. For a more detailed discussion, see  \cite[Chapter 11, Section 11.9]{hunter2001applied}.

Applying Lemma \ref{lem:gaussian-kernel} to eq.~\eqref{eq:39}, it follows that
\begin{align*}
  Z = \Bigl(\frac{\mu}{\tau}\Bigr)^p e^{\frac1{4\tau}\Delta_C} \prod_{j=1}^p \frac{1}{\mu^2-(w_j+ik_j)^2} \biggr|_{k=0} = \Bigl(\frac{\mu}{\tau}\Bigr)^p \Bigl[\prod_{j=1}^p \frac{1}{\mu^2-w_j^2} + \O\bigl(\frac1\tau\bigl)\Bigr],
\end{align*}
with $\Delta_C$ as defined in eq.~\eqref{eq:40}. It follows that the effect size expectation values are, to first order in $\tau^{-1}$,
\begin{align*}
  \E(x_j) = \frac1{2\tau}\frac{\partial \log Z}{\partial w_j} \sim \frac1\tau \frac{w_j}{\mu^2-w_j^2}, 
\end{align*}
which indeed converge to the minimum-energy solution $\hat x=0$.

\subsection{Generalized partition functions for the expected effects}
\label{sec:gener-part-funct}

Using elementary properties of the Fourier transform, it follows that
\begin{equation}\label{eq:4}
  \F\bigl[x_j e^{-2\tau f(x)}\bigr](k) = i \frac{\partial \F\bigl[ e^{-2\tau f(x)}\bigr](k)}{\partial k_j} ,
\end{equation}
with $f$ defined in eq.~\eqref{eq:f}, and hence, repeating the calculations leading up to eq.~\eqref{eq:Z}, we find
\begin{equation}\label{eq:33}
  \E(x_j) = \frac{\int_{\R^p} x_j e^{-\tau H(x)}dx}{\int_{\R^p} e^{-\tau H(x)}dx} = \frac{Z\Bigl[ \bigl(C^{-1}(w-z)\bigr)_j\Bigr]}{Z} \sim  \bigl[C^{-1}(w -\ut)\bigr]_j .
\end{equation}
Note that eq.~\eqref{eq:4} can also be applied to the Laplacian part $e^{-2\tau g(x)}$, with $g$ defined in eq.~\eqref{eq:g}. This results in
\begin{equation}\label{eq:34}
  \E(x_j) =  \frac{Z\Bigl[ \frac{z_j}{\tau(\mu^2-z_j^2)}\Bigr]}{Z} \sim  \frac{\utj}{\tau(\mu^2-\utj^2)}.
\end{equation}
By the saddle point equations, eq.~\eqref{eq:03}, eqs.~\eqref{eq:33} and \eqref{eq:34} are identical. As a rule of thumb, `tricks' such as eq.~\eqref{eq:4} to express properties of the posterior distribution as generalized partition functions lead to accurate approximations if the final result does not depend on whether the trick was applied to the Gaussian or Laplacian part of the Gibbs factor. For higher-order moments of the posterior distribution, this means that the leading term of the stationary phase approximation alone is not sufficient.

\section{Proof of Theorem \ref{thm:main}}
\label{sec:proof-theorem}

\subsection{Saddle-point equations}
\label{sec:saddle-point-equat}

Consider the function $\Htau$ defined in eq.~\eqref{eq:htau},
\begin{align*}
  \Htau(z) = (z-w)^TC^{-1} (z-w) -\frac1\tau\sum_{j=1}^p  \ln(\mu^2-z_j^2),  
\end{align*}
with $z$ restricted to the domain $\D=\{z\in\C^p\colon |\Re z_j|<\mu,\; j=1,\dots,p\}$. Writing $z=u+iv$, where $u$ and $v$ are the real and imaginary parts of $z$, respectively, we obtain
\begin{align*}
  \Re \Htau(z) &= (u-w)^TC^{-1} (u-w) - v^TC^{-1} v - \frac1{2\tau}\sum_{j=1}^p \Bigl\{ \ln\bigl[ (\mu+u_j)^2+v_j^2)\bigr] + \ln\bigl[ (\mu-u_j)^2+v_j^2)\bigr]\Bigr\}\\
  \Im \Htau(z) &= 2 (u-w)^T C^{-1} v - \frac1\tau\sum_{j=1}^p \Bigl\{ \arctan\bigl(\frac{v_j}{\mu+ u_j}\bigr) + \arctan\bigl(\frac{v_j}{\mu- u_j}\bigr) \Bigr\},
\end{align*}
where $\Re c$ and $\Im c$ denote the real and imaginary parts of a complex number $c$, respectively.

By the Cauchy-Riemann equations $z=u+iv$ is a saddle point of $\Htau$ if and only if it satisfies the equations
\begin{align*}
  \frac{\partial \Re \Htau}{\partial u_j} &= 2[C^{-1}(u-w)]_j - \frac1{\tau}\Big\{ \frac{\mu+u_j}{(\mu+u_j)^2+v_j^2} -  \frac{\mu-u_j}{(\mu-u_j)^2+v_j^2}\Bigr\}=0\\
  \frac{\partial \Re \Htau}{\partial v_j} &= -2[C^{-1}v]_j - \frac1{\tau}\Big\{ \frac{v_j}{(\mu+u_j)^2+v_j^2} +  \frac{v_j}{(\mu-u_j)^2+v_j^2}\Bigr\}=0
\end{align*}
The second set of equations is solved by $v=0$, and because $\Re \Htau(u+iv)<\Re \Htau(u)$ for all $u$ and $v\neq 0$, it follows that $v=0$ is the saddle point solution. Plugging this into the first set of equations gives
\begin{equation}\label{eq:3}
  [C^{-1}(u-w)]_j + \frac{u_j}{\tau(\mu^2-u_j^2)} = 0,
\end{equation}
which is equivalent to eq.~\eqref{eq:03}.

\subsection{Analytic expression for the partition function}
\label{sec:analyt-expr-part}

Next, consider the complex integral
\begin{align*}
  \I = (-i)^p\int_{-i\infty}^{i\infty}\dots \int_{-i\infty}^{i\infty} e^{\tau \Htau(z)} Q(z) dz_1\dots dz_p,
\end{align*}
i.e. $\I$ is the generalized partition function upto a constant multiplicative factor. By Cauchy's theorem we can freely deform the integration contours to a set of vertical contours running parallel to the imaginary axis and  passing through the saddle point, i.e.\ integrate over $z=\ut + ik$, where $\ut$ is the saddle point solution and $k\in\R^p$. Changing the integration variable back from complex $z$ to real $k$, we find
\begin{align*}
  \I =  e^{\tau(w-\ut)C^{-1}(w-\ut)}\int_{\R^p} e^{-\tau F(k)} Q(\ut+ik) dk
\end{align*}
where
\begin{align*}
  F(k) &= k^T C^{-1} k - 2ik^T C^{-1} (\ut - w) + \frac1\tau \sum_{j=1}^p
  \ln(\mu - \utj - ik_j)  + \frac1\tau \sum_{j=1}^p  \ln(\mu + \utj + ik_j).
\end{align*}
We start by computing the Taylor series for $F$. First note that the $n^{\text{th}}$ derivative of
$f^\pm_j(k_j)=\ln(\mu \pm \utj \pm ik_j)$ evaluated at $k_j=0$ is given by
\begin{align*}
  (f_j^{\pm})^{(n)}(0) = -\frac{(\mp i)^n (n-1)!}{(\mu \pm \utj)^n}.
\end{align*}
By the saddle point equations \eqref{eq:3}
\begin{align*}
  \frac1\tau \sum_{j=1}^p f^{+'}_j(0) k_j + \frac1\tau \sum_{j=1}^p
  f^{-'}_j(0) k_j= \frac{i}\tau \sum_{j=1}^p 
  \frac{k_j}{\mu + \utj} -\frac{i}\tau \sum_{j=1}^p 
  \frac{k_j}{\mu - \utj} = 2i k^T C^{-1}(\utj -w).
\end{align*}
Hence the linear terms cancel and we obtain
\begin{align*}
  F(k) &= \frac1\tau \sum_{j=1}^p \bigl[ \ln(\mu + \utj) +
  \ln(\mu - \utj) \bigr]+ k^T C^{-1}  k + \frac1{\tau} \sum_{j=1}^p \frac{\mu^2+\utj^2}{(\mu^2 - \utj^2)^2}  k_j^2 \\
  &\qquad - \frac1\tau \sum_{j=1}^p \sum_{n\geq 3} \frac1n \Bigl[\frac{1}{(\mu -\utj)^n} + \frac{(-1)^n}{(\mu +  \utj)^n} \Bigr] (ik_j)^n\\
  &= \frac1\tau \sum_{j=1}^p \ln(\mu^2 - \utj^2) +k^T (C^{-1}+D_\tau^{-1}) k  - \frac1\tau R_\tau(ik),
\end{align*}
with $D_\tau$ the diagonal matrix defined in eq.~\eqref{eq:Dtau} and $R_\tau$ the function defined in eq.~\eqref{eq:Rtau}. Hence
\begin{align*}
  \I =  e^{\tau(w-\ut)C^{-1}(w-\ut)} \prod_{j=1}^p \frac1{\mu^2 - \utj^2} \int_{\R^p} e^{-\tau k^T (C^{-1}+D_\tau^{-1}) k} e^{R_\tau(ik)} Q(\ut+ik) dk. 
\end{align*}
Application of Lemma \ref{lem:gaussian-kernel} results in
\begin{align*}
  &\int_{\R^p} e^{-\tau k^T (C^{-1}+D_\tau^{-1}) k} e^{R_\tau(ik)} Q(\ut+ik) dk \\
  &\qquad= \frac{(2\pi)^{\frac{p}2}}{(2\tau)^{\frac{p}2} \sqrt{\det (C^{-1}+D_\tau^{-1})}} \exp\Bigl\{\frac 1{4\tau^2}\Delta_\tau\Bigr\} e^{R_\tau(ik)} Q(\ut+ik)\biggr|_{k=0}\\
  &\qquad= \Bigl(\frac\pi\tau\Bigr)^{\frac{p}2} \Bigl(\frac{\det(D_\tau)\det(C)}{\det(C+D_\tau)}\Bigr)^{\frac12} \exp\Bigl\{\frac 1{4\tau^2}\Delta_\tau\Bigr\} e^{R_\tau(ik)} Q(\ut+ik)\biggr|_{k=0}\\
  &\qquad=\pi^{\frac{p}2} \frac{\prod_j (\mu^2-\utj^2)}{\prod_j(\mu^2+\utj^2)^{\frac12}} \Bigl(\frac{\det(C)}{\det(C+D_\tau)}\Bigr)^{\frac12} \exp\Bigl\{\frac 1{4\tau^2}\Delta_\tau\Bigr\} e^{R_\tau(ik)} Q(\ut+ik)\biggr|_{k=0},
\end{align*}
where we used the equality
\begin{align*}
  C^{-1} + D_\tau^{-1} = C^{-1} (C+D_\tau) D_\tau^{-1},
\end{align*}
and $\Delta_\tau$ is the differential operator defined in eq.~\eqref{eq:Delta-tau}. Hence
\begin{align*}
  Z[Q] &= \frac{\mu^p}{(\pi\tau)^{\frac{p}2}\sqrt{\det(C)}} \I \\
   &= \Bigl(\frac{\mu}{\sqrt\tau}\Bigr)^p \frac{1}{\prod_j(\mu^2+\utj^2)^{\frac12}} \frac{e^{\tau(w-\ut)C^{-1}(w-\ut)}}{\sqrt{\det(C+D_\tau)}} \exp\Bigl\{\frac 1{4\tau^2}\Delta_\tau\Bigr\} e^{R_\tau(ik)} Q(\ut+ik)\biggr|_{k=0}.
\end{align*}
The derivation above is formal and meant to illustrate how the various terms in the partition function approximation arise. It is rigorous if the inverse Fourier transform of $e^{R_\tau(ik)}Q(\ut+ik)$ exists at least a a tempered distribution (cf.\ the proof of Lemma \ref{lem:gaussian-kernel}). This is the case if $Q$ has compact support. If this is not the case, one first has to truncate $\I$ to a compact region around the saddle point, and use standard estimates \citep{wong2001asymptotic} that the contribution of the region not containing the saddle point is exponentially vanishing. Likewise, application of the operator $e^{\frac 1{4\tau^2}\Delta_\tau}$ is defined through its series expansion, but this is to be understood as an asymptotic expansion (see below) which need not result in a convergent series. None of this is different from the standard theory for the asymptotic approximation of integrals \citep{wong2001asymptotic}.

\subsection{Asymptotic properties of the saddle point}
\label{sec:asympt-prop-saddle}

Let $\uh=\lim_{\tau\to\infty} \ut$. By continuity, $\uh$ is a solution to the set of equations
\begin{equation}\label{eq:5}
  (u_j-\mu)(u_j+\mu)\bigl[C^{-1}(u-w)\bigr]_j = 0
\end{equation}
subject to the constraints $|u_j|\leq\mu$. Denote by $I\subseteq\{1,\dots,p\}$ the subset of indices $j$ for which $\bigl[C^{-1}(\uh-w)\bigr]_j\neq 0$. To facilitate notation, for $v\in\R^p$ a vector, denote by $v_I\in\R^{|I|}$ the sub-vector corresponding to the indices in $I$. Likewise denote by $C_I\in\R^{|I|\times|I|}$ the corresponding sub-matrix and by $C_I^{-1}$ the inverse of $C_I$, i.e.  $C_I^{-1} = (C_I)^{-1}\neq (C^{-1})_I$. Temporarily denoting $B=C^{-1}$, we can then rewrite the equations for $\uh$  as
\begin{align*}
  \uh_{I} &= \pm\mu\\
  \bigl[C^{-1}(\uh-w)\bigr]_{I^c} &= [B(\uh-w)]_{I^c} =  B_{I^c}(\uh_{I^c}-w_{I^c}) + B_{I^c I}(\uh_{I}-w_{I}) =0,
\end{align*}
or, using standard results for the inverse of a partitioned matrix \citep{horn1985},
\begin{align*}
  \uh_{I^c} = w_{I^c} + B_{I^c}^{-1}B_{I^cI} (w_{I}-\uh_{I})
  =w_{I^c} - C_{I^cI}C_{I}^{-1}(w_{I}-\uh_{I}).
\end{align*}
Finally, define $\xh = C^{-1}(w-\uh)$, and note that
\begin{align}
  \hat x_I &  = [B(w-\uh)]_I = B_I (w_I-\uh_I) + B_{II^c}(w_{I^c} - \uh_{I^c}) 
             = (B_I -  B_{II^c} B_{I^c}^{-1}B_{I^c,I} ) (w_I-\uh_I) \nonumber\\
           & =  C_I^{-1}(w_I-\uh_I)\neq 0 \label{eq:20}\\
  \hat x_{I^c} &= 0. \label{eq:21}
\end{align}
As we will see below, $\hat x=\argmin_{x\in\R^p} H(x)$ is the maximum-likelihood lasso or elastic net solution (cf. Appendix~\ref{sec:zero-temp-limit}), and hence the set $I$ corresponds to the set of non-zero coordinates in this solution. Note that it is possible to have $\hat u_j=\pm\mu$ for $j\in I^c$ (i.e. $\hat x_j=0$). This happens when $\mu$ is exactly at the transition value where $j$ goes from not being included to being included in the ML solution. We will denote the subsets of $I^c$ of transition and non-transition coordinates as $I^c_t$ and $I^c_{nt}$, respectively. We then have the following lemma:
\begin{lem}\label{lem:u-lim}
  In the limit $\tau\to\infty$, we have
  \begin{equation}\label{eq:8}
    \tau (\mu^2-\utj^2)^2 = 
    \begin{cases}
      \O(\tau^{-1}) & j\in I\\
      \O\big[(\tau \xtj^2)^{-1}\bigr] & j\in I^c_t\\
      \O(\tau) & j\in I^c_{nt}
    \end{cases}
  \end{equation}
\end{lem}
\begin{proof}
  From the saddle point equations, we have
  \begin{align*}
    \tau  (\mu^2-\utj^2)^2 = \frac1\tau \Bigl(\frac{\utj}{\xtj}\Bigr)^2.
  \end{align*}
  If $j\in I$, $\xtj\to\xhj\neq 0$ and $\utj\to\uhj=\pm\mu$, and hence $\tau  (\mu^2-\utj^2)^2 = \O(\tau^{-1})$. If $j\in I^c_{nt}$, $\mu^2-\utj^2\to \mu^2-\uhj^2>0$, and hence $\tau  (\mu^2-\utj^2)^2 =\O(\tau)$.  If $j\in I^c_{t}$,  $\xtj\to 0$ and $\utj\to\uhj=\pm\mu$, and hence $\tau  (\mu^2-\utj^2)^2 = \O\bigl[(\tau \xtj^2)^{-1}\bigr]$.
\end{proof}

\subsection{Asymptotic properties of the differential operator matrix}
\label{sec:asympt-prop-diff}

Let
\begin{equation}\label{eq:12}
  E_\tau = \tau D_\tau (C+D_\tau)^{-1}C = \frac\tau2\bigl[D_\tau (C+D_\tau)^{-1}C + C (C+D_\tau)^{-1} D_\tau],
\end{equation}
where the second equality is simply to make the symmetry of $E_\tau$ explicit. We have the following result:
\begin{prop}\label{prop:Diff-lim}
  Using the block matrix notation introduced above, and assuming $I^c_t=\emptyset$, the leading term of $E_\tau$ in the limit $\tau\to\infty$ can be written as
  \begin{equation}\label{eq:11}
    E_\tau\sim \tau
    \begin{pmatrix}
      D_{\tau,I} & \frac12 D_{\tau,I} C_I^{-1} C_{II^c}\\
      \frac12 D_{\tau,I} C_I^{-1} C_{II^c} & (C^{-1})_{I^c}
    \end{pmatrix},
  \end{equation}
  where $I$ is again the set of non-zero coordinates in the maximum-likelihood solution.
\end{prop}
\begin{proof}
  Again using standard properties for the inverse of a partitioned matrix \citep{horn1985}, and the fact that $D_\tau$ is a diagonal matrix, we have for any index subset $J$
  \begin{align}
    \bigl[(C+D_\tau)^{-1}\bigr]_J &= \bigl[ C_J + D_{\tau,J} - C_{J,J^c} (C_{J^c} + D_{\tau, J^c})^{-1} C_{J^c,J}\bigr]^{-1} \label{eq:9}\\
     \bigl[(C+D_\tau)^{-1}\bigr]_{J,J^c} &= -(C_J+D_{\tau,J})^{-1} C_{J^c,J} \bigl[(C+D_\tau)^{-1}\bigr]_{J^c}\label{eq:10}
  \end{align}
By Lemma \ref{lem:u-lim}, in the limit $\tau\to\infty$, $D_\tau$ vanishes on $I$ and diverges on $I^c$. Hence
\begin{align}
  (C_I+D_{\tau,I})^{-1} &\sim C_I^{-1} \label{eq:13}\\
  (C_{I^c}+D_{\tau,I^c})^{-1} &\sim  D_{\tau,I^c}^{-1}\label{eq:14}
\end{align}
Plugging these in eqs.~\eqref{eq:9} and \eqref{eq:10}, and using the fact that $C_{I,I^c} D_{\tau,I^c}^{-1} C_{I^c,I}$ is vanishingly small compared to $C_I$, yields
\begin{align*}
  (C+D_\tau)^{-1} \sim
  \begin{pmatrix}
    C_I^{-1}   & -C_I^{-1} C_{I,I^c} D_{\tau,I^c}^{-1} \\
    -D_{\tau,I^c}^{-1}  C_{I^c,I} C_I^{-1} 
    & D_{\tau,I^c}^{-1}
  \end{pmatrix}
\end{align*}
Plugging this in eq.~\eqref{eq:12}, and again using that $D_{\tau,I^c}^{-1}$ is vanishingly small compared to constant matrices yields eq.~\eqref{eq:11}. 
\end{proof}
From the fact that  by Lemma \ref{lem:u-lim}, $\tau D_{\tau,I}\sim\text{const}$, it follows immediately that, if $I^c_t=\emptyset$,
\begin{equation}\label{eq:16}
    (E_\tau)_{ij} =
    \begin{cases}
        \O(\tau) & i,j\in I^c \\
        \text{const} & \text{otherwise}\\
    \end{cases}
  \end{equation}

For transition coordinates, eq.~\eqref{eq:8} may diverge or not, depending on the rate of $\xtj\to 0$. Define
\begin{equation}\label{eq:35}
  J = I \cup \bigl\{ j\in  I^c_t \colon \lim_{\tau\to\infty} \tau^{\frac12} \xtj \neq 0 \bigr\}.
\end{equation}
Then $D_\tau$ diverges on $J^c$ and converges (but not necessarily vanishes) on $J$, and eqs.~\eqref{eq:13} and \eqref{eq:14} remain valid if we use the set $J$ rather than $I$ to partition the matrix (with a small modification in eq. \eqref{eq:13} to keep an extra possible constant term). Hence, we obtain the following modification of eq.~\eqref{eq:16}:
\begin{equation}
  \label{eq:17}
  (E_\tau)_{ij} =
    \begin{cases}
      \O(\tau) & i,j\in J^c\\
      \text{const} & \text{otherwise}
    \end{cases}
\end{equation}

\subsection{Asymptotic properties of the differential operator argument}
\label{sec:asympt-prop-diff-2}

Next we consider the function $R_\tau(z)$ appearing in the argument of the differential operator in eq.~\eqref{eq:Z-anal} and defined in eq.~\eqref{eq:Rtau},
\begin{align*}
  R_\tau(z) &= \sum_{j=1}^p R_{\tau,j}(z_j)\\
  R_{\tau,j}(z_j) &= \sum_{m\geq 3} \frac1m \Bigl[\frac{1}{(\mu - \utj)^m} + \frac{(-1)^m}{(\mu + 
                     \utj)^m} \Bigr] (z_j)^m.
\end{align*}
We have the following result:
\begin{lem}\label{lem:Rtau}
  $R_{\tau,j}(z_j)$ is of the form
  \begin{align*}
    R_{\tau,j}(z_j) = z_j^3 q_{\tau,j}(z_j)
  \end{align*}
  with $q_{\tau,j}$ an analytic function in a region around $z_j=0$ and
  \begin{align*}
    q_{\tau,j} (z_j) \leq
    \begin{cases}
      \O(\tau^2) & j\in J\\
      \O(\tau) & j\in J^c \cap I^c_t\\
      \text{const} & j\in I^c_{nt}
    \end{cases}
  \end{align*}
  with $J$ defined in eq.~\eqref{eq:35}. 
\end{lem}
\begin{proof}
  The first statement follows from the fact that the series expansion of $R_{\tau,j}(z_j)$ contains only powers of $z_j$ greater than 3. The asymptotics as a function of $\tau$ for $j\in I$ and $j\in I^c_{nt}$ follow immediately from Lemma \ref{lem:u-lim} and the definition of $R_{\tau,j}$ (Appendix \ref{sec:analyt-expr-part}),
  \begin{align*}
    R_{\tau,j}(z_j)   &=  -\ln\bigl[\mu^2 - (\utj + z_j)^2\bigr] + \ln(\mu^2 - \utj^2) - \frac{2 \utj}{\mu^2 -  \utj^2} z_j -  \frac{\mu^2+\utj^2}{(\mu^2 - \utj^2)^2} z_j^2.
  \end{align*}
  For $j\in J\cap I^c_t$, we have from Lemma \ref{lem:u-lim} at worst $(\mu^2 - \utj^2)^{-2}=\O\bigl[(\tau \xtj)^{2}\bigr]\leq \O(\tau^2)$, whereas for $j\in J^c\cap I^c_t$, we have at worst $(\tau \xtj)^{2} = \tau (\tau^{\frac12} \xtj)^{2} \leq \O(\tau)$.
\end{proof}

\subsection{Asymptotic approximation for the partition function}
\label{sec:asympt-appr-part}

To prove the analytic approximation eq.~\eqref{eq:Z-approx}, we will show that successive terms in the series expansion of $e^{\frac1{4\tau^2}\Delta_\tau}$ result in terms of decreasing power in $\tau$. The argument presented below is identical to existing proofs of the stationary phase approximation for multi-dimensional integrals \citep{wong2001asymptotic}, except that we need to track and estimate the dependence on $\tau$ in both $\Delta_\tau$ and $R_\tau$. 

The series expansion of the differential operator exponential can be written as:
\begin{align*}
  \exp\Bigl\{\frac1{4\tau^2}\Delta_\tau\Bigr\} 
  &= \sum_{m\geq0} \frac{1}{m!(2\tau)^{2m}} \Delta_\tau^m \\
  &= \sum_{m\geq0} \frac{1}{m!(2\tau)^{2m}}
    \sum_{j_1,\dots,j_{2m}=1}^p E_{j_1j_2}\dots E_{j_{2m-1}j_{2m}} \frac{\partial^{2m}}{\partial
    k_{j_1}\dots \partial k_{j_{2m}}} \\
  &= \sum_{m\geq0} \frac{1}{m!(2\tau)^{2m}} \sum_{\alpha\colon
    |\alpha|=2m} S_{\tau,\alpha} \frac{\partial^{2m}}{\partial
    k_1^{\alpha_1}\dots \partial k_p^{\alpha_p}},
\end{align*}
where $E$ is the matrix defined in eq.~\eqref{eq:12} (its dependence on $\tau$ is omitted for notational simplicity), $\alpha=(\alpha_1,\dots,\alpha_p)$ is a multi-index, $|\alpha|=\sum_j \alpha_j$, and $S_{\tau,\alpha}$ is the sum of all terms $E_{j_1j_2}\dots E_{j_{2m-1}j_{2m}}$ that give rise to the same multi-index $\alpha$. From eq.~\eqref{eq:17}, it follows that only coordinates in $J^c$ give rise to diverging terms in $S_{\tau,\alpha}$, and only if they are coupled to other coordinates in $J^c$. Hence the total number $\sum_{j\in J^c} \alpha_j$ of $J^c$ coordinates can be divided over at most $\frac12\sum_{j\in J^c} \alpha_j$ $E$-factors, and we have
\begin{align*}
  S_{\tau,\alpha} \leq \O\bigl(\tau^{\frac12\sum_{j\in J^c} \alpha_j}\bigr).
\end{align*}

Turning our attention to the partial derivatives, we may assume without loss of generality  that the argument function  $Q$ is a finite sum of products of monomials and hence it is sufficient to prove eq.~\eqref{eq:Z-approx} with $Q$ of the form $Q(z)=\prod_{j=1}^p Q_j(z_j)$. By Cauchy's theorem and Lemma \ref{lem:Rtau}, we have for $\epsilon>0$ small enough,
\begin{align*}
  \frac{\partial^{\alpha_j}}{\partial k_j^{\alpha_j}}
  e^{R_{\tau,j}(ik_j)}Q_j(ik_j)\Bigr|_{k_j=0} &= \frac{\alpha_j!}{2\pi i} \oint_{|z|=\epsilon}
  \frac1{z^{\alpha_j+1}} e^{R_{\tau,j}(z_j)}Q_j(z_j) dz_j\\
  &= \frac{\alpha_j!}{2\pi i} \sum_{n\geq 0} \frac{1}{n!}
    \oint_{|z|=\epsilon} z_j^{3n-\alpha_j-1}  q_j(z_j)^n Q_j(z_j) dz_j \\
  &=  \frac{\alpha_j!}{2\pi i} \sum_{0\leq n< \frac13(\alpha_j+1)} \frac{1}{n!}
    \oint_{|z|=\epsilon} z_j^{3n-\alpha_j-1}  q_j(z_j)^n Q_j(z_j) dz\\
  &\leq
    \begin{cases}
      \O\bigl(\tau^{\frac23\alpha_j} \bigr)&  j\in J\\
      \O\bigl(\tau^{\frac13\alpha_j} \bigr)& j\in J^c\cap I^c_t\\
      \text{const} & j\in I^c_{nt}
    \end{cases}
\end{align*}
The last result follows, because for $j\in J$ or $j\in  J^c\cap I^c_t$, $q_j$ scales at worst as $\tau^2$ or $\tau$, respectively, and hence, since only powers of $q_j$ strictly less than $\frac13(\alpha_j+1)$ contribute to the sum, the sum must be a polynomial in $\tau$ of degree less than $\frac23\alpha_j$ or $\frac13\alpha_j$,  respectively ($\alpha_j$ can be written as either $3t$, $3t+1$ or $3t+2$ for some integer $t$; in all three cases, the largest integer strictly below $\frac13(\alpha_j+1)$ equals $t$, and $t\leq\frac13\alpha_j$). 

Hence
\begin{align*}
  & \sum_{\alpha\colon
  |\alpha|=2m} S_{\tau,\alpha} \frac{\partial^{2m}}{\partial  k_1^{\alpha_1}\dots \partial
    k_p^{\alpha_p}} e^{R_\tau(ik)} Q(ik)\biggr|_{k=0}
    = \sum_{\alpha\colon
    |\alpha|=2m} S_{\tau,\alpha} \prod_j \frac{\partial^{\alpha_j} }{\partial k_j^{\alpha_j}}
    e^{R_{\tau,j}(ik_j)}Q_j(ik_j) \Bigr|_{k_j=0}\\
  &\qquad\leq\O\bigl(\tau^{\frac12\sum_{j\in J^c} \alpha_j}
    \tau^{\frac23\sum_{j\in J}\alpha_j + \frac13 \sum_{j\in J^c\cap I^c_t}\alpha_j}\bigr)
    = \O\bigl(\tau^{\frac23\sum_{j\in J}\alpha_j + \frac12\sum_{j\in I^c_{nt}} \alpha_j+ \frac56 \sum_{j\in J^c\cap I^c_t}\alpha_j}\bigr)\\
  &\qquad \leq \O\bigl(\tau^{\frac56\sum_{j=1}^p\alpha_j}\bigr) 
    = \O\bigl(\tau^{\frac53 m}\bigr)
\end{align*}
This in turn implies that the $m^{\text{th}}$ term in the expansion,
\begin{equation}\label{eq:18}
  \exp\Bigl\{\frac1{4\tau^2}\Delta_\tau\Bigr\} e^{R_\tau(ik)} Q(ik)\biggr|_{k=0}
  = \sum_{m\geq0} \frac{1}{m!(2\tau)^{2m}} \sum_{\alpha\colon
    |\alpha|=2m} S_{\tau,\alpha} \prod_j \frac{\partial^{\alpha_j}
    }{\partial k_j^{\alpha_j}} e^{R_{\tau,j}(ik_j)}Q_j(ik_j)  \Bigr|_{k_j=0} 
\end{equation}
is bounded by a factor of $\tau^{-\frac13 m}$. Hence eq. \eqref{eq:18}
is an asymptotic expansion, with leading term
\begin{align*}
  \exp\Bigl\{\frac1{4\tau^2}\Delta_\tau\Bigr\} e^{R_\tau(ik)} Q(ik) \biggr|_{k=0} \sim \prod_{j=1}^p Q_j(0) = Q(0).
\end{align*}
\qed

\section{Zero-temperature limit of the partition function}
\label{sec:zero-temp-limit}

The connection between the analytic approximation \eqref{eq:Z-approx} and the minimum-energy (or maximum-likelihood) solution is established by first recalling that Fenchel's convex duality theorem implies that \citep{michoel2016}
\begin{align*}
  \hat x = \argmin_{x\in\R^p} H(x) = \argmin_{x\in\R^p}\bigl[ f(x) + g(x) \bigr] = \nabla f^*(-\hat u) = C^{-1}(w-\hat u), 
\end{align*}
where $f$ and $g$ are defined in eqs.~\eqref{eq:f}--\eqref{eq:g},
\begin{align*}
  f^*(u) = \max_{x\in\R^p} \bigl[ x^Tu - f(x) \bigr] =\frac12 (w+u)^TC^{-1}(w+u)
\end{align*}
is the Legendre-Fenchel transform of $f$, and
\begin{equation}\label{eq:uh-2}
  \hat u = \argmin_{\{u\in\R^p\colon |u_j|\leq\mu, \forall j\}} f^*(-u) = \argmin_{\{u\in\R^p\colon |u_j|\leq\mu, \forall j\}} (w-u)^TC^{-1}(w-u).
\end{equation}

One way of solving an optimization problem with constraints of the form $|u_j|\leq\mu$ is to approximate the hard constraints by a smooth, so-called `logarithmic barrier function' \citep{boyd2004convex}, i.e. solve the unconstrained problem
\begin{equation}\label{eq:utau-2}
  \hat u_\tau = \argmin_{u\in\R^p}\Bigl[ (w-u)^TC^{-1}(w-u) - \frac1\tau \sum_{j=1}^p \ln(\mu^2-u_j^2)\Bigr]
\end{equation}
such that in the limit $\tau\to\infty$, $\hat u_\tau\to \hat u$. Comparison with eqs.~\eqref{eq:htau}--\eqref{eq:03}, shows that \eqref{eq:utau-2} is precisely the saddle point of the partition function, whereas the constrained optimization in eq.~\eqref{eq:uh-2} was already encountered in eq.~\eqref{eq:5}. Hence, let $I$ again denote the set of non-zero coordinates in the maximum-likelihood solution $\xh$. The following result characterizes completely the partition function in the limit $\tau\to\infty$, provided there are no transition coordinates.
\begin{prop}\label{prop:1}
  Assume that $\mu$ is not a transition value, i.e.\ $j\in I\Leftrightarrow \xhj\neq 0 \Leftrightarrow |\uhj|=\mu$.  Let $\sigma=\sgn(\hat u)$ be the vector of signs of $\hat u$.  Then $\sgn(\hat x_I)=\sigma_I$, and
  \begin{equation}\label{eq:23}
    Z \sim
    \frac{e^{\tau
        (w_{I}-\mu\sigma_{I})^TC_{I}^{-1}   (w_{I}-\mu\sigma_{I})}}{2^{\frac{|I|}2}\tau^{\frac{|I|}2+|I^c|} 
      \sqrt{\det(C_I)}} \prod_{j\in I^c} \frac{\mu}{\mu^2-\hat u_{j}^2} .
  \end{equation}
 In particular,
  \begin{align*}
    \lim_{\tau\to\infty} \frac1\tau \ln Z =  (w_I-\mu \sigma_I )^T C_I^{-1}  (w_I-\mu\sigma_I) = H(\hat x) = \min_{x\in\R^p} H(x). 
\end{align*}
\end{prop}
\begin{proof}
  First note that from the saddle point equations
  \begin{align*}
    (\mu^2-\utj^2) \xtj = \frac{\utj}\tau,
  \end{align*}
  where as before $ \xt=C^{-1}(w-\ut)$, and the fact that $|\utj|<\mu$, it follows that $\sgn(\xtj)=\sgn(\utj)$ for all $j$ and all $\tau$. Let $j\in I$. Because $\xtj\to \xhj\neq 0$, it follows that there exists $\tau_0$ large enough such that $\sgn(\xtj)=\sgn(\xhj)$ for all $\tau>\tau_0$. Hence also $\sgn(\utj)=\sgn(\xhj)$ for all $\tau>\tau_0$, and since $\utj\to\uhj\neq 0$, we must have $\sgn(\uhj)=\sgn(\xhj)$.

To prove eq.~\eqref{eq:23}, we will calculate the leading term of $\det(C+D_\tau)$ in eq.~\eqref{eq:Z-approx}. For this purpose, recall that for a square matrix $M$ and
any index subset $I$, we have \citep{horn1985}
\begin{equation}\label{eq:19}
  \det(M) = \det(M_I) \det(M_{I^c}-M_{I^cI}M_I^{-1}M_{II^c}) = \frac{\det(M_I)}{\det\bigl[ (M^{-1}) _{I^c} \bigr]}
\end{equation}
Taking $M=C+D_\tau$, it follows from eqs.~\eqref{eq:9}--\eqref{eq:14} that $\det(C_I+D_{\tau,I})\sim \det(C_I)$, and $\det\bigl[ (M^{-1}) _{I^c} \bigr] \sim \det(D_{\tau,I^c}^{-1})$, and hence
\begin{align*}
  \det(C+D_\tau) \sim \det(C_I) \det(D_{\tau,I^c}) = \tau^{|I^c|} \det(C_I) \prod_{j\in I^c} \frac{(\mu^2-\utj^2)^2}{\mu^2+\utj^2}.
\end{align*}
Hence
\begin{align*}
  \tau^{\frac{p}2} \prod_{j=1}^p \sqrt{\mu^2+\utj^2} \sqrt{\det(C+D_\tau)}
  &\sim \tau^{\frac{p+|I^c|}2}\sqrt{\det(C_I)}   \prod_{j\in I} \sqrt{\mu^2+\utj^2} \prod_{j\in I^c} (\mu^2-\utj^2)\\
  &\sim \tau^{\frac{p+|I^c|}2} 2^{\frac{|I|}2} \mu^{|I|} \sqrt{\det(C_I)}  \prod_{j\in I^c} (\mu^2-\uhj^2),
\end{align*}
where the last line follows by replacing $\utj$ by its leading term $\uhj$, and using $\uhj^2=\mu^2$ for $j\in I$. Plugging this in eq.~\eqref{eq:Z-approx} and using eqs.~\eqref{eq:20}--\eqref{eq:21} to get the leading term of the exponential factor results in eq.~\eqref{eq:23}.
\end{proof}

The leading term in eq.~\eqref{eq:23} has a pleasing interpretation as a `two-phase' system,
\begin{align*}
  Z = \frac1{(2\pi)^{\frac{|I|}2}} Z_I Z_{I^c}
\end{align*}
where $Z_I$ and $Z_{I^c}$ are the partition functions (normalization constants) of a multivariate Gaussian distribution and a product of independent shifted Laplace distributions, respectively:
\begin{align*}
  Z_I &= \bigl(\frac\pi\tau\bigr)^{\frac{|I|}2}
        \frac{ e^{\tau (w_I-\mu\sigma_I)^TC_I^{-1}
        (w_I-\mu\sigma_I)}}{\sqrt{\det(C_I)}} = \int_{\R^{|I|}} e^{-\tau[ 
        x_I^T C_I  x_I- 2 (w_I-\mu\sigma_I)^T x_I]}    dx_I \\
  Z_{I^c} &= \frac1{\tau^{|I^c|}}\prod_{j\in I^c}\frac{\mu}{\mu^2 -
            \hat u_j^2} = \int_{\R^{|I^c|}} e^{-2\tau[\mu\sum_{j\in I^c} |x_j| -
            \hat u_{I^c}^T x_{I^c} ]} dx_{I^c}.
\end{align*}
This suggests that in the limit $\tau\to\infty$, the non-zero maximum-likelihood coordinates are approximately normally distributed and decoupled from the zero coordinates, which each follow a shifted Laplace distribution. At finite values of $\tau$ however, this approximation is too crude, and more accurate results are obtained using the leading term of eq.~\eqref{eq:Z-approx}. This is immediately clear from the fact that the partition function is a continous function of $w\in\R^p$, which remains true for the leading term of eq.~\eqref{eq:Z-approx}, but not for eq.~\eqref{eq:23}, which exhibits discontinuities whenever a coordinate enters or leaves the set $I$ as $w$ is smoothly varied.

\section{Analytic results for independent predictors}
\label{sec:analyt-results-indep}

When predictors are independent, the matrix $C$ is diagonal, and the partition function can be written as a product of one-dimensional integrals
\begin{equation*}
  Z = \int_\R e^{-\tau(cx^2 - 2wx + 2\mu|x|)} dx,
\end{equation*}
where $c,\mu>0$ and $w\in\R$. This integral can be solved by writing $Z=Z^++Z^-$, where
\begin{multline}\label{eq:24}
  Z^\pm = \int_0^\infty e^{-\tau[cx^2 \pm 
  2(w\pm\mu)x]} dx = e^{\tau\frac{(w\pm \mu)^2}{c}}
  \int_0^\infty  e^{-\tau c(x\pm\frac{w\pm\mu}{c})^2} dx =
  \frac{ e^{\tau\frac{(w\pm \mu)^2}{c}}}{\sqrt{\tau c}} 
  \int_{\pm\sqrt{\frac{\tau}{c}}(w\pm\mu)}^\infty  e^{-  y^2} dy \\
        =  \frac12 \sqrt{\frac{\pi}{\tau c}} e^{\tau\frac{(w\pm\mu)^2}{c}} \erfc\Bigl(\pm\sqrt{\frac{\tau}{c}}(w\pm\mu)\Bigr) 
   = \frac12 \sqrt{\frac{\pi}{\tau c}} \erfcx\Bigl(\pm\sqrt{\frac{\tau}{c}}(w\pm\mu)\Bigr),
\end{multline}
where $\erfc(x)=\frac2{\sqrt\pi}\int_x^\infty e^{-y^2}dy$ and $\erfcx(x)=e^{x^2}\erfc(x)$ are the complementary and scaled complementary error functions, respectively. Hence,
\begin{align*}
  \log Z &= \log\Bigl[ \erfcx\Bigl(\sqrt{\frac{\tau}{c}}(\mu+w)\Bigr) + \erfcx\Bigl(\sqrt{\frac{\tau}{c}}(\mu-w)\Bigr)\Bigr] + \frac12\bigl(\log\pi-\log(\tau c)\bigr) - \log 2 ,
\end{align*}
and
\begin{align*}
  \xt &= \E(x) = \frac1{2\tau}\frac{\partial \log Z}{\partial w} \\
      &= \frac1{c} \frac{(\mu+w) \erfcx\bigl(\sqrt{\frac{\tau}{c}}(\mu+w)\bigr) - (\mu-w) \erfcx\bigl(\sqrt{\frac{\tau}{c}}(\mu-w)\bigr)}{\erfcx\bigl(\sqrt{\frac{\tau}{c}}(\mu+w)\bigr) + \erfcx\bigl(\sqrt{\frac{\tau}{c}}(\mu-w)\bigr)}\\
  &= \frac{w}{c} +\frac\mu{c}\;\frac{\erfcx\bigl(\sqrt{\frac{\tau}{c}}(\mu+w)\bigr) - \erfcx\bigl(\sqrt{\frac{\tau}{c}}(\mu-w)\bigr)}{\erfcx\bigl(\sqrt{\frac{\tau}{c}}(\mu+w)\bigr) + \erfcx\bigl(\sqrt{\frac{\tau}{c}}(\mu-w)\bigr)}\\
  &= \frac{w}{c} + (1-2\alpha) \frac\mu{c},
\end{align*}
where
\begin{equation*}
  \alpha = \frac1{1+\frac{\erfcx\bigl(\sqrt{\frac\tau{c}}(\mu-w) \bigr)}{\erfcx\bigl( \sqrt{\frac\tau{c}}(\mu+w)\bigr)}}.
\end{equation*}

\section{Numerical recipes}
\label{sec:numerical-recipes}

\subsection{Solving the saddle point equations}
\label{sec:solv-set-quadr}

To calculate the partition function and posterior distribution at any value of $\tau$, we need to solve the set of equations in eq. \eqref{eq:03}. To avoid having to calculate the inverse matrix $C^{-1}$, we make a change of variables $x=C^{-1}(w-u)$, or $u=w-Cx$, such that eq. \eqref{eq:03} becomes 
\begin{equation}\label{eq:15}
  x_j \bigl[w_j-(Cx)_j-\mu\bigr] \bigl[w_j-(Cx)_j+\mu\bigr] +
  \frac1\tau \bigl[w_j-(Cx)_j\bigr] = 0.
\end{equation}
We will use a coordinate descent algorithm where one coordinate of $x$ is updated at a time, using the current estimates $\hat x$ for the other coordinates. Defining
\begin{align*}
  a_j = w_j-\sum_{k\neq j} C_{kj} \hat x_k,
\end{align*}
we can write eq. \eqref{eq:15} as
\begin{align*}
  C_{jj}^2 x_j^3 - 2a_jC_{jj} x_j^2
  +\bigl(a_j^2-\mu^2-\frac{C_{jj}}{\tau}\bigr)x_j + \frac{a_j}{\tau}
  = 0
\end{align*}
The roots of this 3rd order polynomial are easily obtained numerically, and by construction there will be a unique root for which $u_j= w_j-(Cx)_j = a_j- C_{jj}x_j$ is located in the interval $(-\mu,\mu)$. This root will be the new estimate $\hat x_j$. Given a new $\hat x_j\nw$, we can update the vector $a$ as
\begin{align*}
  a_k\nw &=
  \begin{cases}
    a_j\od & k=j\\
    a_k\od - C_{kj} \bigl(\hat x_j\nw - \hat x_j\od\bigr) & k\neq j
  \end{cases}
\end{align*}
and proceed to update the next coordinate.

After all coordinates of $\hat x$ have converged, we obtain $\ut$ by performing the matrix-vector operation
\begin{align*}
  \ut = w- C\hat x,
\end{align*}
or, if we only need the expectation values,
\begin{align*}
  \E_\tau(x) = \hat x.
\end{align*}

For $\tau=\infty$, the solution to eq. \eqref{eq:15} is given by the maximum-likelihood effect size vector (cf.\ Appendix \ref{sec:zero-temp-limit}), for which ultra-fast algorithms exploiting the sparsity of the solution are available \citep{friedman2010regularization}. Hence we use this vector as the initial vector for the coordinate descent algorithm for $\tau<\infty$ and expect fast convergence if $\tau$ is large. Solutions for multiple values of $\tau$ can be obtained along a descending path of $\tau$-values, each time taking the previous solution as the initial vector for finding the next solution.

\subsection{High-dimensional determinants in the partition function}
\label{sec:high-dimens-determ}

Calculating the stationary phase approximation to the partition function involves the computation of the $p$-dimensional determinant $\det(C+D_\tau)$ [cf. eq.~\eqref{eq:Z-approx}], which can become computationally expensive in high-dimensional settings. However, when $C$ is of the form $C=\frac{A^TK^{-1}A}{2n}+\lambda\U$ [cf.\ eq.~\eqref{eq:1}] with $A\in\R^{n\times p}$, $K\in\R^{n\times n}$ invertible, and $p>n$, these determinants can be written as $n$-dimensional determinants, using the matrix determinant lemma:
\begin{equation}\label{eq:7}
  \det(C+D_\tau) = \det\Bigl(\frac{A^TK^{-1}A}{2n}+D'_\tau\Bigr)
  =\frac{\det(D'_\tau)}{\det(K)} \det\Bigl(K+\frac{A(D_\tau')^{-1}A^T}{2n}\Bigr),
\end{equation}
where $D_\tau' = D_\tau+\lambda\U$ is a diagonal matrix whose determinant and inverse are trivial to obtain.

To avoid numerical overflow or underflow, all calculations are performed using logarithms of partition functions. For $n$ large, a numerically stable computation of eq.~\eqref{eq:7} uses the equality $\log\det B = \tr\log B = \sum_{i=1}^n \log \epsilon_i$, where $B=K+\frac1{2n}A(D_\tau')^{-1}A^T$ and $\epsilon_i$ are the eigenvalues of $B$.

\subsection{Marginal posterior distributions}
\label{sec:marg-post-distr}

Calculating the marginal posterior distributions $p(x_j)$ [eq.~\ref{eq:marginal}] requires applying the analytic approximation eq.~\eqref{eq:Z-anal} using a different $\ut$ for every different value of $x_j$. To make this process more efficient, two simple properties are exploited:
\begin{enumerate}
\item For $x_j=\xtj$, the saddle point for the $(p-1)$-dimensional partition function $Z(C_{I_j},w_{I_j}-x_j C_{j,I_j},\mu)$ is given by the original saddle point vector $\hat x_{\tau,k}$, $k\neq j$. This follows easily from the saddle point equations.
\item If $x_j$ changes by a small amount, the new saddle point also changes by a small amount. Hence, taking the current saddle point vector for $x_j$ as the starting vector for solving the set of saddle point equations for the next value $x_j+\delta$ results in rapid convergence (often in a single loop over all coordinates).
\end{enumerate}
Hence we always start by computing $p(x_j=\xtj)$ and then compute $p(x_j)$ separately for a series of ascending values $x_j>\xtj$ and a series of descending values $x_j<\xtj$

\subsection{Sampling from the one-dimensional distribution}
\label{sec:sampling-from-one}

Consider again the case of one predictor, with posterior distribution
\begin{equation}\label{eq:26}
  p(x) = \frac{e^{-\tau(cx^2-2wx+2\mu|x|)}}{Z}.
\end{equation}
To sample from this distribution, note that
\begin{align*}
  p(x) = (1-\alpha)\, p(x\mid x<0)  + \alpha\, p(x\mid x\geq0),
\end{align*}
where
\begin{equation}\label{eq:25}
  p(x\mid x\in\R^{\pm}) = \frac{e^{-\tau(cx^2-2(w\mp\mu)x)}}{Z^{\mp}},
\end{equation}
$Z^\pm$ were defined in eq.~\eqref{eq:24}, and
\begin{align*}
 \alpha &= P(x\geq 0) = \int_0^\infty p(x) dx= \frac1{Z}
       \int_0^\infty e^{-\tau[cx^2 -  2(w-\mu)x]} dx = \frac{Z^{-}}{Z}
  = \frac1{1+\frac{\erfcx\bigl(\sqrt{\frac\tau{c}}(\mu-w) \bigr)}{\erfcx\bigl( \sqrt{\frac\tau{c}}(\mu+w)\bigr)}}.
\end{align*}

Eq.~\eqref{eq:25} defines two truncated normal distributions with means $(w\mp\mu)/c$ and standard deviation $1/\sqrt{2\tau c}$, for which sampling functions are available. Hence, to sample from the distribution \eqref{eq:26}, we first sample a Bernoulli random variable with probability $\alpha$, and then sample from the appropriate truncated normal distribution.

\subsection{Gibbs sampler}
\label{sec:gibbs-sampl-comp}

To sample from the Gibbs distribution in the general case, we use the `basic Gibbs sampler' of \cite{hans2011elastic}. Let $\hat x$ be the current vector of sampled regression coefficients. Then a new coefficient $x_j$ is sampled from the conditional distribution 
\begin{equation}\label{eq:29}
  p\bigl(x_j\mid \{\hat x_k, k\neq j\} \bigr)= \frac{ e^{-\tau[C_{jj}x_j^2 -
      2a_jx_j +2\mu|x_j|]}}{Z_j},
\end{equation}
where $a_j = w_j-\sum_{k\neq j}C_{kj}\hat x_k$ and $Z_j$ is a normalization constant. This distribution is of the same form as eq.~\eqref{eq:26} and hence can be sampled from in the same way. Notice that, as in section \ref{sec:solv-set-quadr}, after sampling a new $\hat x_j$, we can update the vector $a$ as 
\begin{align*}
  a_k\nw &=
  \begin{cases}
    a_j\od & k=j\\
    a_k\od - C_{kj} \bigl(\hat x_j\nw - \hat x_j\od\bigr) & k\neq j
  \end{cases}.
\end{align*}

\subsection{Maximum a-posteriori estimation of the inverse temperature}
\label{sec:hyperp-estim}

 This paper is concerned with the problem of obtaining the posterior regression coefficient distribution for the Bayesian lasso and elastic net when values for the hyperparameters $(\lambda,\mu,\tau)$ are given. There is abundant literature on how to select values for $\lambda$ and $\mu$ for maximum-likelihood estimation, mainly through cross validation or by predetermining a specific level of sparsity (i.e.\ number of non-zero predictors). Hence we assume an appropriate choice for $\lambda$ and $\mu$ has been made, and propose to then set $\tau$ equal to a first-order approximation of its maximum a posteriori (MAP) value, i.e.\ finding the value which maximizes the log-likelihood of observing data $y\in\R^n$ and $A\in\R^p$, similar to what was suggested by \cite{hans2011elastic}. To do so we must include the normalization constants in the prior distributions \eqref{eq:31}--\eqref{eq:32}:
\begin{align*}
  p(y\mid A,x,\tau) &= \Bigl(\frac{\tau}{2\pi n}\Bigr)^{\frac{n}2}  e^{-\frac{\tau}{2n}\|y-Ax\|^2} = \Bigl(\frac{\tau}{2\pi n}\Bigr)^{\frac{n}2} e^{-\frac{\tau}{2n}\|y\|^2} e^{-\frac\tau{2n} [ x^TA^TAx - 2(A^Ty)^Tx]}\\
  p(x\mid \lambda,\mu,\tau) &= \frac{e^{-\tau(\lambda\|x\|^2+2\mu\sum_j|x_j|)}}{Z_0} 
\end{align*}
where for $\lambda>0$,
\begin{multline}\label{eq:27}
   Z_0 = \int_{\R^p}dx\; e^{-\tau(\lambda\|x\|^2+2\mu\sum_j|x_j|)} = \Bigl( \int_\R dx\; e^{-\tau(\lambda x^2 + 2\mu|x|)}\Bigr)^p = \Bigl(2 \int_{0}^\infty dx\; e^{-\tau(\lambda x^2+2\mu x)} \Bigr)^p\\
  = \Bigl( \frac{2e^{\frac{\mu^2\tau}{\lambda}}}{\sqrt{\lambda\tau}} \int_{\sqrt\frac{\mu^2\tau}{\lambda}}^\infty e^{-t^2} dt\Bigr)^p = \biggl( \sqrt\frac{\pi}{\lambda\tau} e^{\frac{\mu^2\tau}{\lambda}} \erfc\Bigl(\sqrt\frac{\mu^2\tau}{\lambda}\Bigr)\biggr)^p \sim \Bigl(\frac1{\mu\tau}\Bigr)^p,
\end{multline}
and the last relation follows from the first-order term in the asymptotic expansion of the complementary error function for large values of its argument,
\begin{align*}
  \erfc(x) \sim \frac{e^{-x^2}}{x\sqrt\pi}.
\end{align*}
For pure lasso regression ($\lambda=0$), this relation is exact:
\begin{align*}
  Z_0 &= \Bigl(\frac1{\mu\tau}\Bigr)^p.
\end{align*}
Hence, the log-likelihood of observing data $y\in\R^n$ and $A\in\R^p$ given values for $\lambda$, $\mu$, $\tau$ is
\begin{align*}
  \L &= \log \int_{\R^p}dx\, p(y\mid A,x,\tau) p(x\mid \lambda,\mu,\tau) \\
     &= \frac{n}2\log\tau - \frac{\|y\|^2}{2n}\tau - \log Z_0 + \log \int_{\R^p} dx\; e^{-\tau H(x)} + \text{const},
\end{align*}
where `const' are constant terms not involving the hyperparameters. Taking the first order approximation
\begin{align*}
  \log Z = \log \int_{\R^p} dx\; e^{-\tau H(x)} \sim -\tau H_{\min} = -\tau H(\xh),
\end{align*}
where $\xh$ are the maximum-likelihood regression coefficients, we obtain
\begin{align*}
  \L &\sim \bigl(p+\frac{n}2\bigr)\log\tau - \Bigl[\frac{\|y\|^2}{2n} + H(\xh)\Bigr]\tau +p\log\mu\\
  &= \bigl(p+\frac{n}2\bigr)\log\tau - \Bigl[\frac{1}{2n}\|y-A\xh\|^2 + \lambda\|\xh\|^2 + 2\mu\|\xh\|_1\Bigr]\tau +p\log\mu
\end{align*}
which is maximized at
\begin{align*}
  \tau = \frac{p+\frac{n}2}{\frac{1}{2n}\|y-A\xh\|^2 + \lambda\|\xh\|^2 + 2\mu\|\xh\|_1}.
\end{align*}
Note that a similar approach to determine the MAP value for $\lambda$ would require keeping an additional second order term in eq.~\eqref{eq:27}, and that for $p>n$ it is not possible to simultaneously determine MAP values for all three hyperparameters, because it leads to a set of equations that are solved by the combination $\lambda=\mu=0$ and $\tau=\infty$.

\section{Experimental details}
\label{sec:experimental-details}

\subsection{Hardware and software}
\label{sec:hardware-software}

All numerical experiments were performed on a standard Macbook Pro with 2.8 GHz processor amd 16 GB RAM running macOS version 10.13.6 and Matlab version R2018a. Maximum-likelihood elastic net models were fitted using Glmnet for Matlab (\url{https://web.stanford.edu/~hastie/glmnet_matlab/}). Matlab software to solve the saddle point equations, compute the partition function and marginal posterior distributions, and run a Gibbs sampler, is available at \url{https://github.com/tmichoel/bayonet/}. Bayesian horseshoe and an alternative Bayesian lasso Gibbs sampler were run using the BayesReg toolbox for Matlab \cite{makalic2016high}, available at
\url{https://uk.mathworks.com/matlabcentral/fileexchange/60823-bayesian-penalized-regression-with-continuous-shrinkage-prior-densities}.

\subsection{Diabetes and leukemia data}
\label{sec:diab-leuk-data}

The diabetes data were obtained from \url{https://web.stanford.edu/~hastie/CASI_files/DATA/diabetes.html}. The leukemia data were obtained from \url{https://web.stanford.edu/~hastie/CASI_files/DATA/leukemia.html}. Data were standardized according to eq.~\eqref{eq:stand}, and no further processing was performed. For the results in Figure \ref{fig:marginal-dist}, $\lambda$ was set to $0.1$, $\mu$ was  selected as the smallest value with a maximum-likelihood solution with 5 (diabetes data) or 10 (leukemia data) non-zero predictors, and $\tau$ was set to its maximum a-posteriori value given $\lambda$ and $\mu$, yielding $\tau=682.3$ (diabetes data) and $9.9439\cdot 10^3$ (leukemia data).

\subsection{Cancer Cell Line Encyclopedia data}
\label{sec:cencer-cell-line}

Normalized expression data for 18,926 genes in 917 cancer cell lines were obtained from the Gene Expression Omnibus accession number GSE36139 using the Series Matrix File \texttt{GSE36139-GPL15308\_series\_matrix.txt}. Drug sensitivity data for 24 compounds in 504 cell lines were obtained from the supplementary material of \cite{barretina2012cancer} (tab 11 from supplementary file \texttt{nature11003-s3.xls}); 474 cell lines were common between gene expression and drug response data and used for our analyses. Of the available drug response data, only the activity area (`actarea') variable was used; 7 compounds had more than 40 zero activity area values (meaning inactive compounds) in the 474 cell lines and were discarded. For the remaining 17 compounds, the following procedure was used to compare the stationary phase approximation for the Bayesian elastic net to BayReg's lasso and horseshoe regression methods and maximum-likelihood elastic net and ridge regression: 

\begin{enumerate}
\item For each response variable (drug), possible hyper-parameter values were set to $\lambda=0.1$ (fixed); $\mu_n=\mu_{\max} \times r^{\frac{N+1-n}{N}}$, where $N=10$, $n=1,\dots,10$, $r=0.01$ and $\mu_{\max}=\max_{j=1,\dots,p} |w_j|$, with $w$ as defined in eq.~\eqref{eq:1}--\eqref{eq:2} and $p=18,926$; $\tau_m=10^{0.25(m+M-1)}$, where $M=12$, $m=1,2,\dots,13$.
\item For each training data set, and for each drug, the following procedure was performed:
  \begin{enumerate}
  \item The 1,000 genes most strongly correlated with the response were selected as candidate predictors.
  \item Response and predictor data were standardized.
  \item Ten-fold cross-validation was carried out, by randomly dividing the 427 training samples in 10 sets of 42 samples for testing; running Bayesian and maximum-likelihood elastic net on the remaining 385 samples and evaluating predictive performance (correlation of predicted and true drug sensitivities) on the test set. The values $\hat\mu_{ML}$, $\hat\mu_{BAY}$ nd $\hat\tau_{BAY}$ with best median performance were selected.
  \item Maximum-likelihood coefficients for ridge regression ($\lambda=0.1, \mu=0$) and elastic net regression ($\lambda=0.1,\mu=\hat\mu_{ML}$) were calculated on the training set (427 samples).
  \item Bayesian posterior expectation values using the stationary phase approximation for elastic net ($\lambda=0.1,\mu=\hat\mu_{ML},\tau=\hat\tau_{BAY}$)  were calculated on the training set (427 samples).
  \item Bayesian posterior expectation values for lasso and horseshoe regression using Gibbs sampling (using BayReg with default parameter settings) were calculated on the training set (427 samples).
  \item Drug responses were predicted on the original data scale in the 47 held-out validation samples using all sets of regression coefficients, and the Pearson correlation with the true drug response was calculated.
  \end{enumerate}
\item For each drug, the median correlation value over the 10 predictions was taken, resulting in the values shown in Figure~\ref{fig:ccle}a.
\end{enumerate}

The top 1,000 most correlated genes were pre-filtered in each training data set, partly because in trial runs this resulted in better predictive performance than pre-selecting 5,000 or 10,000 genes, and partly to speed up calculations.

Figure~\ref{fig:ccle}d shows the regression coefficients for the training fold whose performance was closest to the median.

For Figure~\ref{fig:ccle}c, Bayesian posterior expectation values using the stationary phase approximation for elastic net and maximum-likelihood regression coefficients were calculated on each training set (427 samples) over a denser grid of 20 values for $\mu$ (same formula as above with $N-20)$ and the same 13 values for $\tau$, and evaluated on the 47 validation samples; the median correlation value of predicted and true drug sensitivities over the ten folds is shown. Figure~\ref{fig:ccle}b shows the dependence on $\mu$ of the maximum-likelihood performance, and performance of the best $\tau$ at every $\mu$ for the Bayesian method.

 \newpage

\section{Supplementary figures}
\label{sec:suppl-figur}

\renewcommand{\thefigure}{S\arabic{figure}}
\setcounter{figure}{0}

\begin{figure}[h!]
  \centering
  \includegraphics[width=\linewidth]{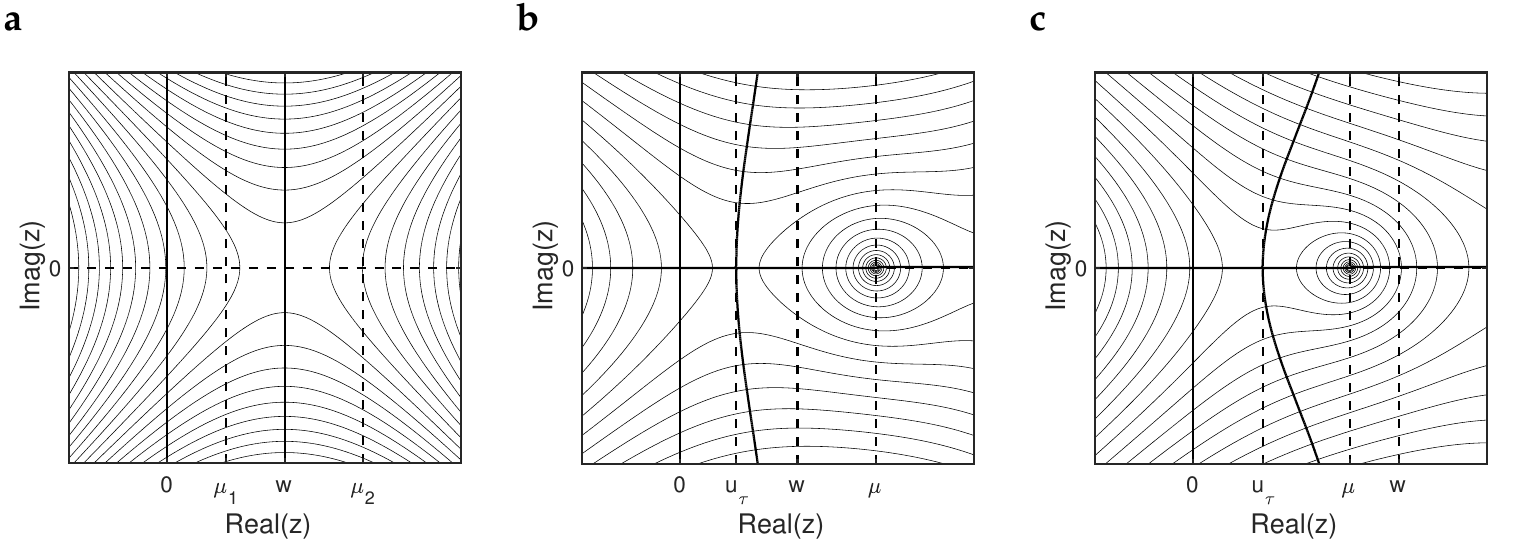}
  \caption{Illustration of the stationary phase approximation procedure for $p=1$. \textbf{(a)} Contour plot of the complex function $(z-w)^2$. If $\mu=\mu_2$, the integration contour can be deformed from the imaginary axis to a steepest descent contour parallel to the imaginary axis and passing through the saddle point $z_0=w$, whereas if $\mu=\mu_1$, this cannot be done without passing through the pole at $z=\mu$. \textbf{(b,c)} Contour plots of the complex function $ (z-w)^2-\frac1\tau\ln(\mu^2-z^2)$ for $|w|<\mu$ and $|w|\geq\mu$, respectively. In both cases the function has a unique saddle point $u_\tau$ with $|u_\tau|<\mu$ and a steepest descent contour that is locally parallel to the imaginary axis.}
  \label{fig:theory}
\end{figure}

\begin{figure}[h!]
  \centering
  \includegraphics[width=\linewidth]{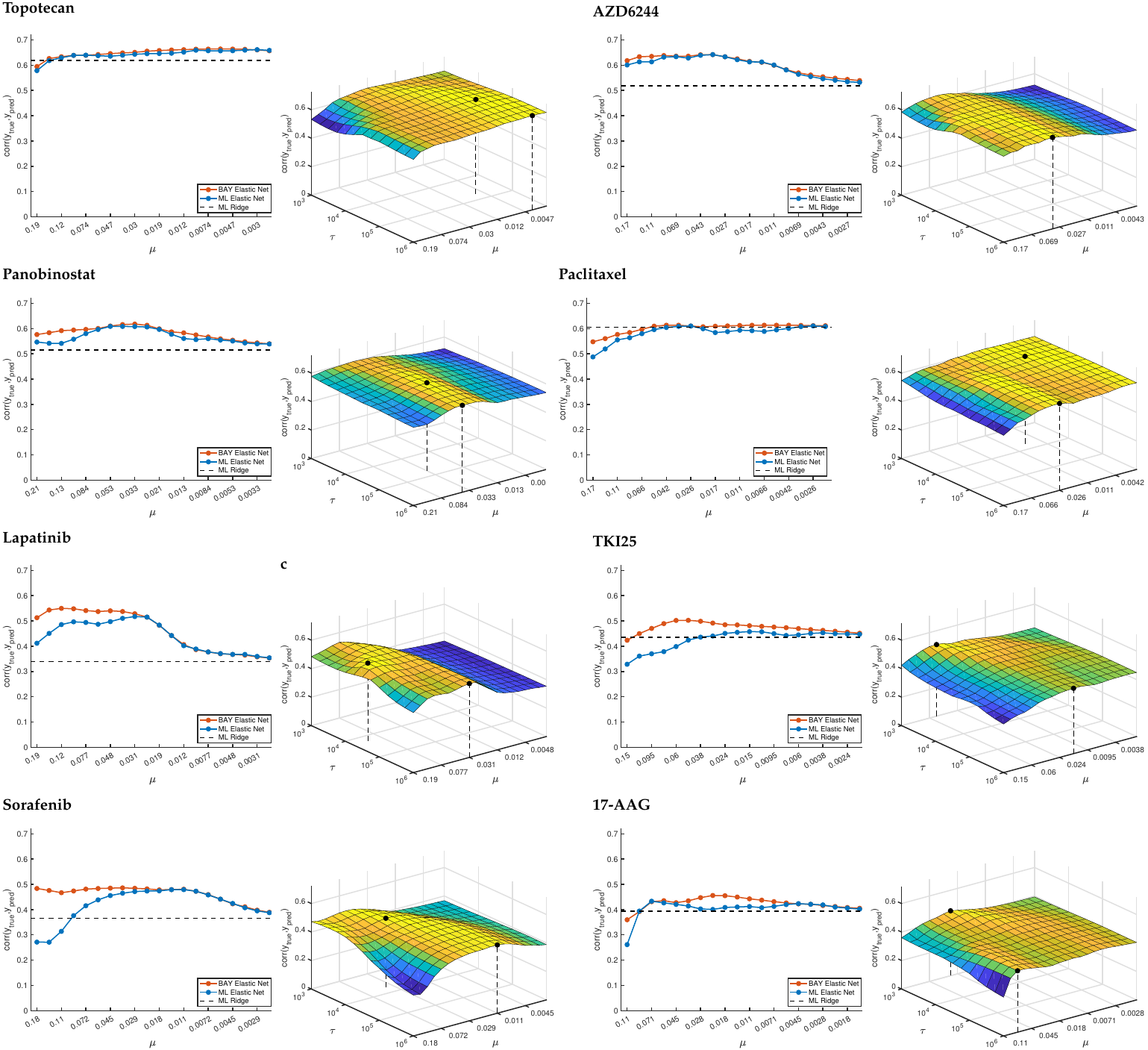}
  \caption{Same as Figure \ref{fig:ccle}b and c, for drugs 2--9 from Figure \ref{fig:ccle}a.}
  \label{fig:ccle-sup-1}
\end{figure}

\begin{figure}
  \centering
  \includegraphics[width=\linewidth]{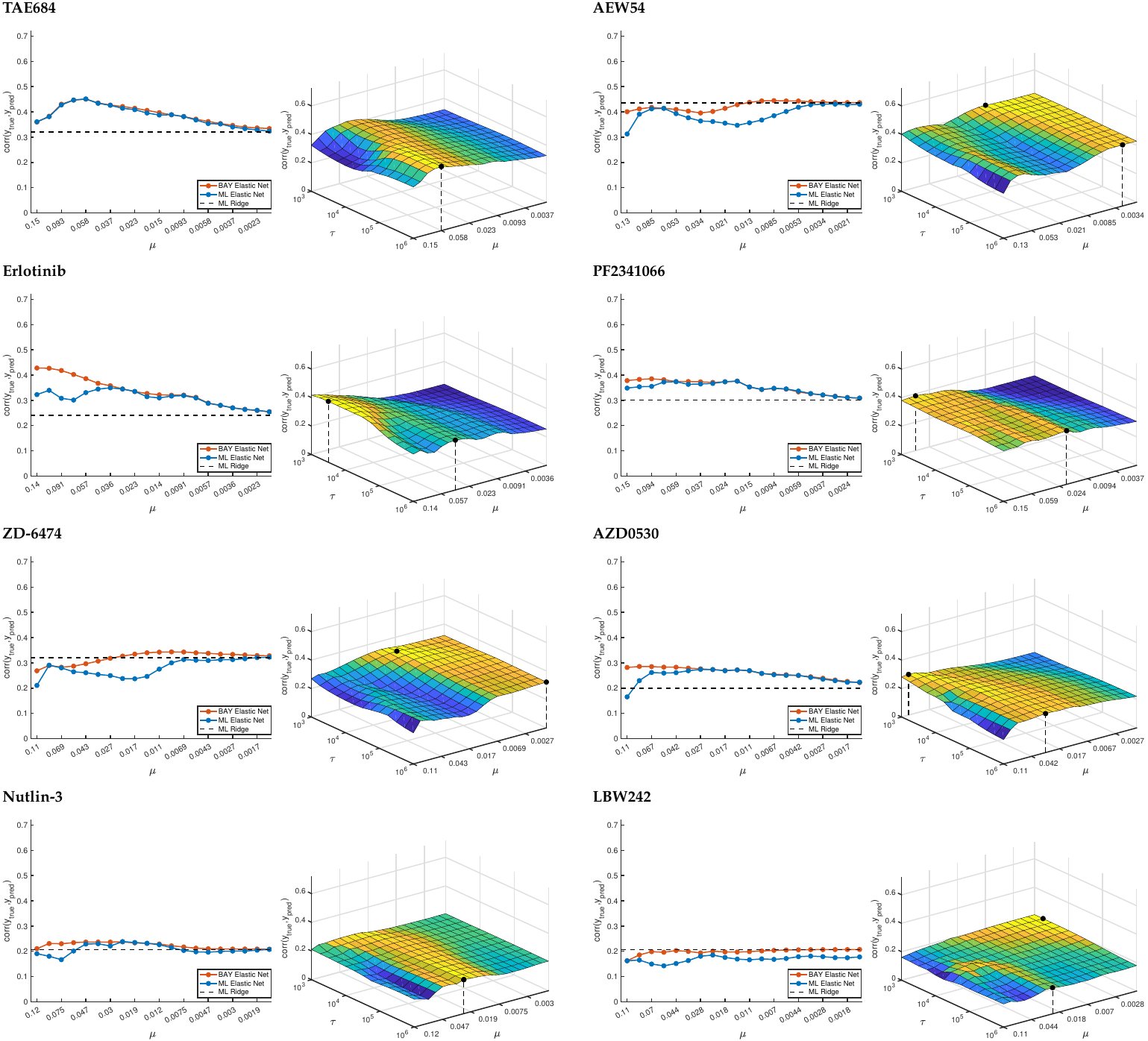}
  \caption{Same as Figure \ref{fig:ccle}b and c, for drugs 10--17 from Figure \ref{fig:ccle}a.}
  \label{fig:ccle-sup-2}
\end{figure}


\end{document}